\pdfoutput=1
\RequirePackage{ifpdf}
\ifpdf 
\documentclass[pdftex]{sigma}
\else
\documentclass{sigma}
\fi

\numberwithin{equation}{section}

\newtheorem{Theorem}{Theorem}[section]
\newtheorem{Corollary}[Theorem]{Corollary}

\newtheorem{Proposition}[Theorem]{Proposition}
 { \theoremstyle{definition}
\newtheorem{Definition}[Theorem]{Definition}}

\usepackage{tikz}

\usepackage{breqn}
\newcommand*\fillbox[1]{\fbox{\begin{minipage}{.975\linewidth} \addtolength{\jot}{0mm} \abovedisplayskip=.5em\belowdisplayskip=.0em {#1} \end{minipage}}}

\renewcommand*{\big}[1]{\scalebox{1.2}{\ensuremath#1}} 

\renewcommand{\d}{\mathrm{d}}
\newcommand{\var}[3]{\frac{\delta_{#1} {#2}}{\delta {#3}}}
\newcommand{\der}[2]{\frac{\partial {#1}}{\partial {#2}}}
\newcommand{\disc}{\mathrm{disc}}

\newcommand{\miwa}{\mathrm{Miwa}}

\newcommand{\C}{\mathbb{C}}
\newcommand{\R}{\mathbb{R}}
\newcommand{\Z}{\mathbb{Z}}
\newcommand{\N}{\mathbb{N}}

\newcommand{\cF}{\mathcal{F}}
\newcommand{\cL}{\mathcal{L}}
\newcommand{\cO}{\mathcal{O}}

\newcommand{\ol}{\overline}

\newcommand{\bn}{\mathbf{n}}
\newcommand{\bt}{\mathbf{t}}

\newcommand{\fe}{\mathfrak{e}}
\newcommand{\fv}{\mathfrak{v}}

\DeclareMathOperator{\Li2}{Li_2}

\newcommand{\parallelo}{
	\begin{tikzpicture}
	\clip (-.2,-.12) rectangle (.2,.2);
	\begin{scope}[cm={.8,0,.8,1,(0,0)}]
	\node[transform shape] at (0,0) {$\blacksquare$};
	\end{scope}
	\end{tikzpicture}
}

\begin{document}
\allowdisplaybreaks

\newcommand{\arXivNumber}{1811.01855}

\renewcommand{\PaperNumber}{044}

\FirstPageHeading

\ShortArticleName{A Variational Perspective on Continuum Limits of ABS and Lattice GD Equations}

\ArticleName{A Variational Perspective on Continuum Limits\\ of ABS and Lattice GD Equations}

\Author{Mats VERMEEREN}

\AuthorNameForHeading{M.~Vermeeren}

\Address{Institut f\"ur Mathematik, MA 7-1, Technische Universit\"at Berlin,\\ Str.~des~17. Juni 136, 10623 Berlin, Germany}
\Email{\href{mailto:vermeeren@math.tu-berlin.de}{vermeeren@math.tu-berlin.de}}
\URLaddress{\url{http://page.math.tu-berlin.de/~vermeer/}}

\ArticleDates{Received November 20, 2018, in final form May 16, 2019; Published online June 03, 2019}

\Abstract{A pluri-Lagrangian structure is an attribute of integrability for lattice equations and for hierarchies of differential equations. It combines the notion of multi-dimensional consistency (in the discrete case) or commutativity of the flows (in the continuous case) with a variational principle. Recently we developed a continuum limit procedure for pluri-Lagrangian systems, which we now apply to most of the ABS list and some members of the lattice Gelfand--Dickey hierarchy. We obtain pluri-Lagrangian structures for many hierarchies of integrable PDEs for which such structures where previously unknown. This includes the Krichever--Novikov hierarchy, the double hierarchy of sine-Gordon and modified KdV equations, and a first example of a continuous multi-component pluri-Lagrangian system.}

\Keywords{continuum limits; pluri-Lagrangian systems; Lagrangian multiforms; multi\-di\-men\-sional consistency}

\Classification{37K10; 39A14}

\section{Introduction}

This paper continues the work started in~\cite{vermeeren2019continuum}, where we established a continuum limit procedure for lattice equations with a \emph{pluri-Lagrangian} (also called \emph{Lagrangian multiform}) structure. Here, we apply this procedure to many more examples. These continuum limits produce known hierarchies of integrable PDEs, but also a pluri-Lagrangian structure for these hierarchies, which in most cases was not previously known.

We start by giving a short introduction to discrete and continuous pluri-Lagrangian systems. Then, in Section~\ref{sec-clim}, we review the continuum limit procedure from~\cite{vermeeren2019continuum}. In Sections~\ref{sec-Q} and~\ref{sec-H} we present examples form the ABS list~\cite{adler2003classification}. In Section~\ref{sec-SG} we extend one of these examples, ABS equation H3, to produce the doubly infinite hierarchy containing the sine-Gordon and modified KdV equations. In Section~\ref{sec-even} we comment on a common feature of all of our continuum limits of ABS equations, namely that half of the continuous independent variables can be disregarded. In Section~\ref{sec-gd} we study two examples from the Gelfand--Dickey hierarchy.

The computations in this paper were performed in the SageMath software system~\cite{sagemath}. The code is available at~\url{https://github.com/mvermeeren/pluri-lagrangian-clim}.

\subsection{The discrete pluri-Lagrangian principle}

Most of the lattice equations we will consider are \emph{quad equations}, difference equations of the form
\begin{gather*} Q(U,U_1,U_2,U_{12},\alpha_1,\alpha_2) = 0, \end{gather*}
where subscripts of the field $U\colon \Z^2 \rightarrow \C$ denote lattice shifts,
\begin{gather*}
U \equiv U(m,n), \qquad U_1 \equiv U(m+1,n), \qquad U_2 \equiv U(m,n+1), \qquad U_{12} \equiv U(m+1,n+1),
\end{gather*}
and $\alpha_i \in \C$ are parameters associated to the lattice directions.

Even though the equations live in $\Z^2$, we require that we can consistently implement them on every square in a higher-dimensional lattice $\Z^N$. This property of \emph{multidimensional consistency} is an important attribute of integrability for lattice equations, see for example~\cite{adler2003classification,boll2014integrability}, or~\cite{hietarinta2016discrete}. A necessary and sufficient condition for multidimensional consistency is that the equation is consistent around the cube:

\begin{figure}[h!]\centering
\includegraphics[scale=0.75]{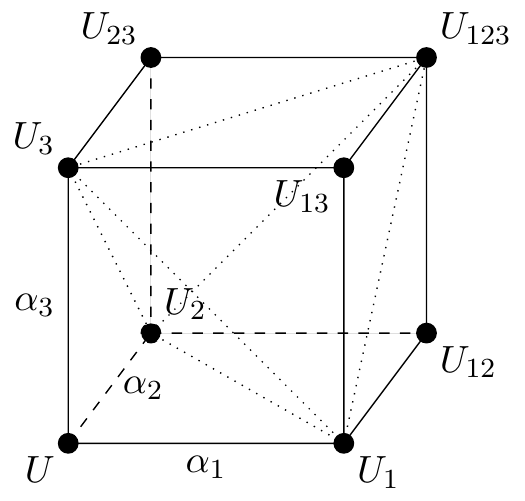}
\caption{A quad equation is consistent around the cube if $U_{123}$ can be uniquely determined from~$U$,~$U_1$, $U_2$ and $U_3$. If in addition $U_{123}$ is independent of~$U$, then the equation satisfies the tetrahedron property.}
\end{figure}

\begin{Definition}Given lattice parameters $\alpha_1$, $\alpha_2$, and $\alpha_3$ and field values $U$, $U_1$, $U_2$, and $U_3$, we can use the equations
\begin{gather*} Q(U,U_i,U_j,U_{ij},\lambda_i,\lambda_j) = 0, \qquad 1 \leq i < j \leq 3\end{gather*}
to determine $U_{12}$, $U_{13}$, and $U_{23}$. Then we can use each of the three equations
\begin{gather*} Q(U_i,U_{ij},U_{ik},U_{ijk},\lambda_j,\lambda_k) = 0, \qquad (i,j,k) \text{ even permutation of } (1,2,3) \end{gather*}
to determine $U_{123}$. If these three values agree for all initial conditions $U$, $U_1$, $U_2$, and $U_3$ and all parameters $\alpha_1$, $\alpha_2$, and $\alpha_3$, then the equation is \emph{consistent around the cube}.
\end{Definition}

Multidimensionally consistent equations on quadrilateral graphs, satisfying certain additional assumptions, were classified by Adler, Bobenko and Suris in~\cite{adler2003classification}. The list of equations they found is widely known as the ABS list. In the same paper Lagrangians were given for each of the equations in the context of the classical variational principle.

We now present the pluri-Lagrangian perspective on these equations, which first appeared in \cite{lobb2009lagrangian} and was explored further in \cite{bobenko2010lagrangian}. Consider the lattice $\Z^N$ with basis vectors $\fe_1,\dots,\fe_N$. To each lattice direction we associate a parameter $\lambda_i \in \C$. We denote an elementary square (a~\emph{quad}) in the lattice by
\begin{gather*} \square_{i,j}(\bn) = \big\{ \bn + \varepsilon_1 \fe_i + \varepsilon_2 \fe_j \,\big|\, \varepsilon_1,\varepsilon_2 \in \{0,1\} \big\} \subset \Z^N, \end{gather*}
where $\bn = (n_1, \dots, n_N)$. Quads are considered to be oriented; interchanging the indices $i$ and $j$ reverses the orientation. We will write $ U(\square_{i,j}(\bn))$ for the quadruple
\begin{gather*} U(\square_{i,j}(\bn)) = \big( U(\bn), U(\bn + \fe_{i}), U(\bn + \fe_{j}), U(\bn + \fe_{i} + \fe_{j}) \big) . \end{gather*}
Occasionally we will also consider the corresponding ``filled-in'' squares in $\R^N$,{\samepage
\begin{gather*} \blacksquare_{i,j}(\bn) = \big\{ \bn + \mu_1 \fe_{i} + \mu_2 \fe_{j} \,\big|\, \mu_1,\mu_2 \in [0,1] \big\} \subset \R^N , \end{gather*}
on which we consider the orientation defined by the basis $(\fe_{i}, \fe_{j})$ of the tangent space.}

The role of a Lagrange function is played by a discrete 2-form
\begin{gather*} L(U(\square_{i,j}(\bn) ), \lambda_{i}, \lambda_{j}), \end{gather*}
which is a function of the values of the field $U\colon \Z^N \rightarrow \C$ on a quad and of the corresponding lattice parameters, satisfying $ L ( U ( \square_{j,i}(\bn)), \lambda_{j}, \lambda_{i} ) = - L( U ( \square_{i,j}(\bn) ), \lambda_{i}, \lambda_{j})$.

\begin{figure}[t]\centering
\includegraphics[scale=0.8]{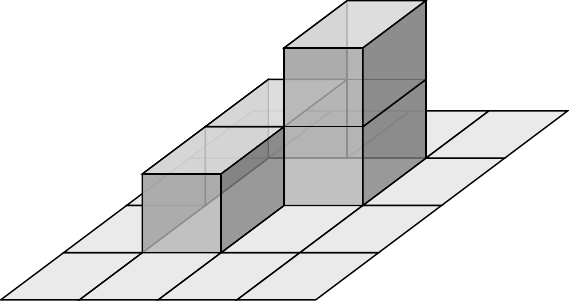}
\caption{Visualization of a discrete $2$-surface in $\mathbb{Z}^3$.}
\end{figure}

Consider a discrete surface $\Gamma = \{\square_\alpha\}$ in the lattice, i.e., a set of quads, such that the union of the corresponding filled-in squares $\bigcup_\alpha \blacksquare_\alpha$ is an oriented topological $2$-manifold (possibly with boundary). The action over $\Gamma$ is given by
\begin{gather*} S_\Gamma = \sum_{\square_{i,j}(\bn) \in \Gamma} L(U(\square_{i,j}(\bn)) ,\lambda_{i}, \lambda_{j}) . \end{gather*}

\begin{Definition}The field $U$ is a solution to the \emph{discrete pluri-Lagrangian problem} if it is a~critical point of $S_\Gamma$ (with respect to variations that are zero on the boundary of $\Gamma$) for all discrete surfaces $\Gamma$ simultaneously.
\end{Definition}

Euler--Lagrange equations in the discrete pluri-Lagrangian setting are obtained by taking point-wise variations of the field on the corners of an elementary cube. Since any other surface can be constructed out of such corners, these variations give us necessary and sufficient conditions.

It is useful to note that exact forms are null Lagrangians.

\begin{Proposition}	\label{prop-exact}	Let $\eta(U,U_i,\lambda_i)$ be a discrete 1-form. Then every field $U\colon \Z^N \rightarrow \C$ is critical for the discrete pluri-Lagrangian $2$-form $L = \Delta \eta$, defined by
	\begin{gather*} L(U,U_i,U_j,U_{ij},\lambda_i,\lambda_j) = \eta(U,U_i,\lambda_i) + \eta(U_i,U_{ij},\lambda_j) - \eta(U_j,U_{ij},\lambda_i) - \eta(U,U_j,\lambda_j) . \end{gather*}
\end{Proposition}
\begin{proof}	The action sum of $L = \Delta \eta$ over a discrete surface $\Gamma$ depends only on the values of~$U$ on the boundary of $\Gamma$. Hence any variation that is zero on the boundary leaves the action invariant.
\end{proof}

For the details of the discrete pluri-Lagrangian theory we refer to the groundbreaking paper~\cite{lobb2009lagrangian}, which introduced the pluri-Lagrangian (or Lagrangian multiform) idea, and to the reviews~\cite{boll2014integrability} and~\cite[Chapter 12]{hietarinta2016discrete}, and the references therein.

\subsection{The continuous pluri-Lagrangian principle}

Continuous pluri-Lagrangian systems are defined analogous to their discrete counterparts. Let $\cL[u]$ be a $2$-form in $\R^N$, depending on a field $u\colon \R^N \rightarrow \C$ and any number of its derivatives.

\begin{Definition}The field $u$ solves the \emph{continuous pluri-Lagrangian problem} for~$\cL$ if for any embedded surface $\Gamma \subset \R^n$ and any variation $\delta u$ which vanishes near the boundary $\partial \Gamma$, there holds
\begin{gather*} \delta \int_\Gamma \cL = 0. \end{gather*}
\end{Definition}

The first question about continuous pluri-Lagrangian systems is to find a set of equations characterizing criticality in the pluri-Lagrangian sense. If $\cL = \sum\limits_{i < j} \cL_{ij}[u] \,\d t_i \wedge \d t_j$, then these equations are
\begin{alignat}{3}
&\label{2el1} \var{ij}{\cL_{ij}}{u_I} = 0 \qquad && \forall I \not\ni t_i,t_j ,& \\
&\label{2el2} \var{ij}{\cL_{ij}}{u_{It_j}} - \var{ik}{\cL_{ik}}{u_{It_k}} = 0 \qquad && \forall I \not\ni t_i ,& \\
&\label{2el3} \var{ij}{\cL_{ij}}{u_{It_it_j}} + \var{jk}{\cL_{jk}}{u_{It_jt_k}} + \var{ki}{\cL_{ki}}{u_{It_kt_i}}= 0\qquad && \forall I ,&
\end{alignat}
for $i$, $j$, and $k$ distinct, where
\begin{gather*} \var{ij}{}{u_I} = \sum_{\alpha,\beta = 0}^\infty (-1)^{\alpha+\beta} \frac{\d^\alpha}{\d t_i^\alpha} \frac{\d^\beta}{\d t_j^\beta} \der{}{u_{I t_i^\alpha t_j^\beta}}, \end{gather*}
and $u_I$ denotes a partial derivative corresponding to the multi-index $I$. We write that $I \ni t_k$ if the entry in $I$ corresponding to $t_k$ is nonzero, i.e., if at least one derivative is taken with respect to $t_k$, and $I \not\ni t_k$ otherwise. Compared to $u_I$, the field $u_{I t_i^\alpha t_j^\beta}$ has $\alpha$ additional derivatives with respect to $t_i$ and $\beta$ additional derivatives with respect to $t_j$.

We call equations \eqref{2el1}--\eqref{2el3} the \emph{multi-time Euler--Lagrange equations}. They were derived in~\cite{suris2016lagrangian}.

\section{Continuum limit procedure}\label{sec-clim}

In this section we review the essentials of the continuum limit procedure for multidimensionally consistent lattice equations and their pluri-Lagrangian structure, presented in~\cite{vermeeren2019continuum}. The continuum limit on the level of the equations has its roots in \cite{miwa1982hirota, wiersma1987lattice}. On the level of the Lagrangian structure, a significant precursor is \cite{yoo2011discrete}, in which the continuum limit of a Lagrangian 1-form system is presented.

\subsection{Miwa variables}\label{sec-miwa}

We construct a map from the lattice $\Z^N(n_1, \dots,n_N)$ to the continuous multi-time $\R^N(t_1,\dots,t_N)$ as follows. We associate a parameter $\lambda_i$ with each lattice direction and set
\begin{gather*} t_i = (-1)^{i+1} \left( n_1 \frac{\lambda_1^i}{i} + \dots + n_N \frac{\lambda_N^i}{i} \right). \end{gather*}
Note that a single step in the lattice (changing one $n_j$) affects all the times $t_i$, hence we are dealing with a very skew embedding of the lattice. We will also consider a slightly more general correspondence,
\begin{gather}\label{miwa}
t_i = (-1)^{i+1} \left(n_1 \frac{c \lambda_1^i}{i} + \dots + n_N \frac{c \lambda_N^i}{i} \right) + \tau_i,
\end{gather}
for constants $c, \tau_1, \dots, \tau_N$ describing a scaling and a shift of the lattice. The variables $n_j$ and~$\lambda_j$ are known in the literature as Miwa variables and have their origin in \cite{miwa1982hirota}. We will call equation~\eqref{miwa} the \emph{Miwa correspondence}.

We denote the shift of $U$ in the $i$-th lattice direction by $U_i$. If $U(\bn) = u(t_1,\dots,t_N)$, It is given by
\begin{gather*}
U_i = U(\bn + \fe_i) =u \left( t_1 + c \lambda_i, t_2 - \frac{c \lambda_i^2}{2}, \dots, t_n - (-1)^N \frac{c \lambda_i^N}{N} \right), \end{gather*}
which we can expand as a power series in $\lambda_i$. The difference equation thus turns into a power series in the lattice parameters. If all goes well, its coefficients will define differential equations that form an integrable hierarchy. Generically however, we will find only $t_1$-derivatives in the leading order, because only in the $t_1$-coordinate do the parameters $\lambda_i$ enter linearly. To get a hierarchy of PDEs, we need to have some leading order cancellation, such that the first nontrivial equation contains derivatives with respect to several times. Given an integrable difference equation, it is a nontrivial task to find an equivalent equation that yields the required leading order cancellation.

Note that such a procedure is strictly speaking not a continuum \emph{limit}: sending $\lambda_i \rightarrow 0$ would only leave the leading order term of the power series. A more precise formulation is that the continuous $u$ interpolates the discrete $U$ for sufficiently small values of $\lambda_i$, where $U$ is defined on a mesh that is embedded in $\R^N$ using the Miwa correspondence. Since $\lambda_i$ is still assumed to be small, it makes sense to think of the outcome as a limit, but it is important to keep in mind that higher order terms should not be disregarded.

\subsection{Continuum limits of Lagrangian forms}

We sketch the limit procedure for pluri-Lagrangian structures introduced in \cite{vermeeren2019continuum}. Consider $N$ pairwise distinct lattice parameters $\lambda_1, \dots, \lambda_N$ and denote by $\fe_1, \dots, \fe_N$ the unit vectors in the lattice $\mathbb{Z}^N$. The differential of the Miwa correspondence maps them to linearly independent vectors in $\R^N$:
\begin{gather*} \fe_i \mapsto \fv_i= \left( c \lambda_i, - \frac{c \lambda_i^2}{2}, \dots , (-1)^{N+1} \frac{c \lambda_i^N}{N} \right) \qquad \text{for } i = 1, \dots, N. \end{gather*}

We start from a discrete Lagrangian two-form $L_\disc$. A discrete field can be recovered from a~continuous field by evaluating it at $u(\bt), u(\bt + \fv_i), u(\bt + \fv_i + \fv_j), \cdots$. We define
\begin{gather*}
\cL_\disc([u], \lambda_1,\lambda_2)= L_\disc\bigg( u, \, u + \partial_1 u + \frac{1}{2} \partial_1^2 u + \cdots,\,
 u + \partial_2 u + \frac{1}{2} \partial_2^2 u + \cdots, \\
\hphantom{\cL_\disc([u], \lambda_1,\lambda_2)= L_\disc\bigg(}{} u + \partial_1 u + \partial_2 u + \frac{1}{2} \partial_1^2 u + \partial_1 \partial_2 u + \frac{1}{2} \partial_2^2 u + \cdots,\ \lambda_1,\lambda_2 \bigg) ,
\end{gather*}
where the differential operators correspond to the lattice directions under the Miwa correspondence,
\begin{gather*} \partial_k = \sum_{j=1}^N (-1)^{j+1} \frac{c \lambda_k^j}{j} \frac{\d}{\d t_j}. \end{gather*}
As the subscript indicates, $\cL_\disc$ is not yet a continuous Lagrangian, because the action is a sum over evaluations of $\cL_\disc$ on a corner of each quad of the surface. Using the Euler--Maclaurin formula we can turn the action sum into an integral of
\begin{gather*}
\cL_\miwa([u],\lambda_1,\lambda_2)
 = \sum_{i,j=0}^\infty \frac{B_{i} B_{j}}{i! j!} \partial_1^{i} \partial_2^{j} \cL_\disc([u],\lambda_1,\lambda_2) ,
\end{gather*}
which is a formal power series in $\lambda_1$ and $\lambda_2$. Then by construction we have the formal equality
\begin{gather*} L_\disc(U(\square_{12}(\bn)),\lambda_1,\lambda_2) = \int_{\parallelo_{12}} \cL_\miwa([u],\lambda_1,\lambda_2) \,\eta_1 \wedge \eta_2 , \end{gather*}
where $(\eta_1, \dots, \eta_N)$ are the 1-forms dual to the Miwa shifts $(\fv_1,\dots,\fv_N)$ and $\parallelo_{12}$ is the embedding under the Miwa correspondence of the filled-in square. Since we want such a relation on arbitrarily oriented quads, we consider the continuous 2-form
\begin{gather*} \cL = \sum_{1 \leq i < j \leq N} \cL_\miwa([u],\lambda_i,\lambda_j) \,\eta_i \wedge \eta_j .\end{gather*}
This is the continuum limit of the discrete pluri-Lagrangian 2-form, but it can be written in a~more convenient form in terms of the coefficients of the power series $\cL_\miwa$.

\begin{Theorem}[\cite{vermeeren2019continuum}]\label{thm-Lmiwa}Let $L_\disc$ be a~discrete Lagrangian $2$-form, such that every term in the corresponding power series $\cL_\miwa$ is of strictly positive degree in both $\lambda_i$, i.e., such that $\cL_\miwa$ is of the form
\begin{gather*}
\cL_\miwa([u],\lambda_1,\lambda_2) = \sum_{i,j=1}^\infty (-1)^{i+j} c^2 \frac{\lambda_1^i}{i} \frac{\lambda_2^j}{j} \cL_{i,j}[u].
\end{gather*}
Then the differential $2$-form
\begin{gather*} \cL = \sum_{1 \leq i < j \leq N} \cL_{i,j}[u] \, \d t_i \wedge \d t_j \end{gather*}
is a pluri-Lagrangian structure for the continuum limit hierarchy, restricted to $\R^N$.
\end{Theorem}

In all examples presented below it was also verified by direct calculation that solutions to the continuum limit equations are indeed critical fields for the corresponding pluri-Lagrangian structure.

Closely related to our definition of a discrete pluri-Lagrangian system is the property that on solutions
\begin{gather}
L_\disc(U(\square_{i,j}(\bn + \fe_k) ), \lambda_{i}, \lambda_{j})- L_\disc(U(\square_{i,j}(\bn) ), \lambda_{i}, \lambda_{j}) \nonumber\\
\qquad{} + L_\disc(U(\square_{j,k}(\bn + \fe_i) ), \lambda_{j}, \lambda_{k})- L_\disc(U(\square_{j,k}(\bn) ), \lambda_{j}, \lambda_{k}) \nonumber\\
\qquad{} + L_\disc(U(\square_{k,i}(\bn + \fe_j) ), \lambda_{k}, \lambda_{i}) - L_\disc(U(\square_{k,i}(\bn) ), \lambda_{k}, \lambda_{i}) = 0,\label{closed}
\end{gather}
which implies that the action sum over any closed surface is zero. Some authors argue that this is the fundamental property of the Lagrangian theory of multi-dimensionally consistent equations \cite{hietarinta2016discrete, lobb2009lagrangian, lobb2010lagrangian, yoo2011discrete}. (The authors that take this perspective mostly use the term ``Lagrangian multiform'' in place of our ``pluri-Lagrangian''.) In the continuum limit, this turns into the property that the action integral vanishes on closed surfaces, when evaluated on solutions. In other words, equation~\eqref{closed} implies that the 2-form $\cL$ found in the continuum limit is closed on solutions of the limit hierarchy.

\subsection{Eliminating alien derivatives}

Suppose a pluri-Lagrangian 2-form in $\R^N$ produces multi-time Euler--Lagrange equations of evolutionary type,
\begin{gather*} u_{t_k} = f_k[u] \qquad \text{for } k \in \{2,3,\dots,N\} . \end{gather*}
If this is the case, then the differential consequences of the multi-time Euler--Lagrange equations can be written in a similar form,
\begin{gather}\label{fullel}
	u_I = f_I[u] \qquad \text{with $I \ni t_k$ for some } k \in \{2,3,\dots,N\},
\end{gather}
where the multi-index $I$ in $f_I$ is a label, not a partial derivative. In this context it is natural to consider $t_1$ as a space coordinate and the others as time coordinates. If the multi-time Euler--Lagrange equations are not evolutionary, equation~\eqref{fullel} still holds for a reduced set of multi-indices $I$.

\begin{Definition}A mixed partial derivative $u_I$ is called \emph{$\{i,j\}$-native} if each individual derivative is taken with respect to $t_i$, $t_j$ or the space coordinate $t_1$, i.e., if
\begin{gather*} I \ni t_k \ \Rightarrow \ k \in \{ 1, i, j \} .\end{gather*}
If $u_I$ is not $\{i,j\}$-native, i.e., if there is a $k \not\in \{ 1,i,j\}$ such that $t_k \in I$, then we say $u_I$ is \emph{$\{i,j\}$-alien}.
	
If it is clear what the relevant indices are, for example when discussing a coefficient $\cL_{i,j}$, we will use \emph{native} and \emph{alien} without mentioning the indices.
\end{Definition}

We would like the coefficient $\cL_{i,j}$ to contain only native derivatives. A naive approach would be to use the multi-time Euler--Lagrange equations \eqref{fullel} to eliminate alien derivatives. Let~$R_{i,j}$ denote the operator that replaces all $\{i,j\}$-alien derivatives for which the multi-time Euler--Lagrange equations provide an expression. We denote the resulting pluri-Lagrangian coefficients by $\ol{\cL}_{i,j} = R_{i,j}(\cL_{i,j})$ and the 2-form with these coefficients by~$\ol{\cL}$. A~priori there is no reason to believe that the 2-form $\ol{\cL}$ will be equivalent to the original pluri-Lagrangian 2-form $\cL$, but using the multi-time Euler--Lagrange equations one can derive the following result~\cite{vermeeren2019continuum}.

\begin{Theorem}\label{thm-alien}If for all $j$ the coefficient $\cL_{1j}$ does not contain any alien derivatives, then every critical field $u$ for the pluri-Lagrangian $2$-form $\cL$ is also critical for $\ol{\cL}$.
\end{Theorem}

In practice, given a Lagrangian 2-form, we can make it fulfill the conditions on the $\cL_{1j}$ by adding a suitable exact form and adding terms that attain a double zero on solutions to the multi-time Euler--Lagrange equations. Neither of these actions affects the multi-time Euler--Lagrange equations.

\section{ABS equations of type Q}\label{sec-Q}

All equations of type Q can be prepared for the continuum limit in the same way, based on their particularly symmetric three leg form. In Section \ref{sec-Q-3-leg} we present this general strategy. In Sections \ref{sec-Q10}--\ref{sec-Q4} we will discuss each individual equation of type Q.

\subsection{Three leg forms and Lagrangians}\label{sec-Q-3-leg}

All quad equations $Q(V,V_1,V_2,V_{12},\lambda_1,\lambda_2) = 0$ from the ABS list have a three leg form:
\begin{gather*} Q(V,V_1,V_2,V_{12},\lambda_1,\lambda_2) = \Psi\big(V,V_1,\lambda_1^2\big) - \Psi\big(V,V_2,\lambda_2^2\big) - \Phi\big(V,V_{12},\lambda_1^2-\lambda_2^2\big).\end{gather*}
(Usually they are written in terms of the parameters $\alpha_i = \lambda_i^2$.)
For the equations of type $Q$, the function $\Phi$ on the long (diagonal) leg is the same as the function $\Psi$ on the short legs:
\begin{gather*} Q(V,V_1,V_2,V_{12},\lambda_1,\lambda_2) = \Psi\big(V,V_1,\lambda_1^2\big) - \Psi\big(V,V_2,\lambda_2^2\big) - \Psi\big(V,V_{12},\lambda_1^2-\lambda_2^2\big). \end{gather*}
Suitable leg functions $\Psi$ were listed in~\cite{bobenko2010lagrangian}. For the purposes of a continuum limit, it is useful to reverse one of the time directions, i.e., to consider
\begin{gather*} Q(V,V_{-1},V_2,V_{-1,2},\lambda_1,\lambda_2) = \Psi\big(V,V_{-1},\lambda_1^2\big) - \Psi\big(V,V_2,\lambda_2^2\big) - \Psi\big(V,V_{-1,2},\lambda_1^2-\lambda_2^2\big), \end{gather*}
and to write the $\Psi$ in terms of difference quotients. We will introduce a function
\begin{gather*} \psi(v,v',\lambda,\mu) = \psi_1(v,\lambda,\mu) + \psi_2 (v',\lambda,\mu) \end{gather*}
 of the continuous variables, from which we can recover $\Psi(V,W,\lambda\mu)$ by plugging in suitable approximations to $v$ and $v'$. Note that $\Psi$ takes only one parameter, which is the product of the two parameters of $\psi$. For all of the ABS equations we will use $c = -2$ in the Miwa correspondence, which means that the derivative $v_{t_1}$ is approximated by difference quotients such as $\frac{V_{-1} - V}{2 \lambda_1}$ and $\frac{V - V_2}{2 \lambda_2}$. We identify
\begin{gather*} \Psi(V,W,\lambda \mu) = \psi \left( \frac{V+W}{2} ,\frac{V-W}{2 \lambda},\lambda,\mu \right) . \end{gather*}

All equations of the Q-list can be written in the form
\begin{gather}
Q(V,V_{-1},V_2,V_{-1,2},\lambda_1,\lambda_2) = \psi \left( \frac{V+V_{-1}}{2} ,\frac{V-V_{-1}}{2 \lambda_1},\lambda_1,\lambda_1 \right) \nonumber \\
\hphantom{Q(V,V_{-1},V_2,V_{-1,2},\lambda_1,\lambda_2) =}{} - \psi\left( \frac{V+V_2}{2} ,\frac{V-V_2}{2 \lambda_2},\lambda_2,\lambda_2 \right)\nonumber \\
\hphantom{Q(V,V_{-1},V_2,V_{-1,2},\lambda_1,\lambda_2) =}{} - \psi \left( \frac{V+V_{-1,2}}{2},\frac{V-V_{-1,2}}{2 (\lambda_1 - \lambda_2)},\lambda_1-\lambda_2,\lambda_1+\lambda_2 \right) .\label{Q-3-leg}
\end{gather}
As suggested by the $D_4$-symmetry of a quad equation, in particular by
\begin{gather*} Q(V,V_{-1},V_2,V_{-1,2},\lambda_1,\lambda_2) = - Q(V,V_2,V_{-1},V_{-1,2},\lambda_2,\lambda_1) ,\end{gather*}
we require that
\begin{gather*}
\psi_1(v,-\lambda,\mu) = -\psi_1(v,\lambda,\mu), \qquad
\psi_2 (-v',-\lambda,\mu) = -\psi_2 (v',\lambda,\mu).
\end{gather*}
Furthermore, we require that
\begin{gather*}\psi_2 (v',-\lambda,\mu) = \psi_2 (v',\lambda,\mu)\end{gather*}
and $\psi(v,v',0,0) = 0$. As we will see below, all ABS equations of type Q have a three-leg form that satisfies these conditions.

We would expect the first nonzero terms in the series expansion at first order in $\lambda_1$, $\lambda_2$, but using the symmetry of $\psi$ one can check that
\begin{gather*} Q(V,V_{-1},V_2,V_{-1,2},\lambda_1,\lambda_2) = \cO \big( \lambda_1^2 + \lambda_2^2 \big) . \end{gather*}
This is the leading order cancellation required to obtain PDEs in the continuum limit: at the first order, where generically we would get only derivatives with respect to $t_1$, we get nothing at all.

Equation \eqref{Q-3-leg} also reveals a reason for considering the three-leg form with a ``downward'' diagonal leg, as in Fig.~\ref{fig-3-leg-form}$(b)$: the difference quotient $\frac{V_{-1,2}-V}{2(\lambda_2 - \lambda_1)}$ can be expanded in a double power series, but its ``upward'' analogue $\frac{V_{1,2}-V}{2(\lambda_1 + \lambda_2)}$ cannot.

To find a Lagrangian for equation~\eqref{Q-3-leg}, we follow \cite{adler2003classification, bobenko2010lagrangian} and integrate the leg function $\psi$. We take
\begin{gather*} \chi_1(v,\lambda,\mu) = \frac{2}{\lambda} \int \psi_1(v,\lambda,\mu) \d v
\qquad \text{and} \qquad
\chi_2(v',\lambda,\mu) = 2 \int \psi_2(v',\lambda,\mu) \d v' . \end{gather*}
Then
\begin{gather*}
\chi_1(v,-\lambda,\mu) = \chi_1(v,\lambda,\mu), \\
\chi_2 (-v',\lambda,\mu) = \chi_2 (v',-\lambda,\mu) = \chi_2 (v',\lambda,\mu),
\end{gather*}
and $\chi = \chi_1 + \chi_2$ satisfies
\begin{gather*} \frac{\lambda}{2} \der{}{v}\chi(v,v',\lambda,\mu) + \frac{1}{2} \der{}{v'}\chi(v,v',\lambda,\mu) = \psi(v,v',\lambda,\mu) . \end{gather*}
Now
\begin{gather*} \Lambda(V,W,\lambda,\mu) = \lambda \, \chi \left( \frac{V+W}{2} ,\frac{V-W}{2\lambda},\lambda,\mu \right) \end{gather*}
gives the terms of the Lagrangian in triangle form:
\begin{gather}
L(V,V_1,V_2,\lambda_1,\lambda_2) = \Lambda(V,V_1,\lambda_1,\lambda_1) - \Lambda(V,V_2,\lambda_2,\lambda_2) - \Lambda(V_1,V_2,\lambda_1 - \lambda_2, \lambda_1+\lambda_2)\nonumber\\
\hphantom{L(V,V_1,V_2,\lambda_1,\lambda_2)}{} = \lambda_1 \chi \left( \frac{V+V_1}{2} ,\frac{V-V_1}{2\lambda_1},\lambda_1,\lambda_1 \right) - \lambda_2 \chi \left( \frac{V+V_2}{2} ,\frac{V-V_2}{2\lambda_2},\lambda_2,\lambda_2 \right) \nonumber\\
\hphantom{L(V,V_1,V_2,\lambda_1,\lambda_2)=}{} - (\lambda_1 - \lambda_2) \chi \left( \frac{V_1+V_2}{2} ,\frac{V_1-V_2}{2(\lambda_1 - \lambda_2)},\lambda_1-\lambda_2,\lambda_1+\lambda_2 \right) .\label{Q-lagrangian}
\end{gather}
Note the symmetries of $\Lambda$:
\begin{gather*} \Lambda(V,W,\lambda,\mu) = \Lambda(W,V,\lambda,\mu) = -\Lambda(V,W,-\lambda,\mu) . \end{gather*}
In some cases we will rescale $\Lambda$ and hence $L$ by a constant factor. This is purely for esthetic reasons and does not affect the multi-time Euler--Lagrange equations.

\begin{figure}[t]\centering
\includegraphics[scale=0.8]{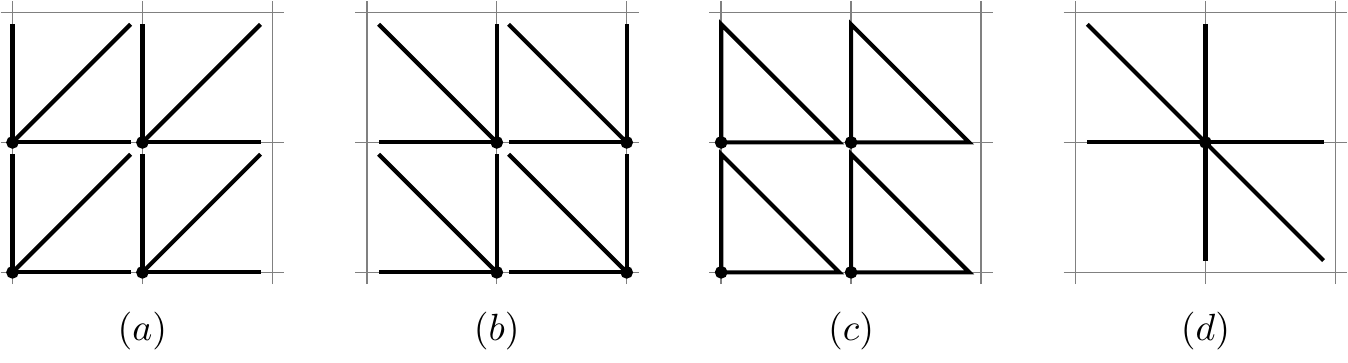}
\caption{The stencils on four adjacent quads for $(a)$ the three-leg form in the usual orientation, $(b)$ the three-leg form after time-reversal, $(c)$ the triangle form for the Lagrangian, and $(d)$ an Euler--Lagrange equation in a planar lattice.}\label{fig-3-leg-form}
\end{figure}

\begin{Proposition}Solutions to the quad equation $Q = 0$ in the plane, with $Q$ given by equation~\eqref{Q-3-leg}, are critical fields for the action of the Lagrangian given by equation~\eqref{Q-lagrangian}.
\end{Proposition}

\begin{proof}We have
\begin{gather*}
\der{}{V} L(V,V_1,V_2,\lambda_1,\lambda_2) = \lambda_1 \der{}{V} \chi \left( \frac{V+V_1}{2} ,\frac{V-V_1}{2\lambda_1},\lambda_1,\lambda_1 \right) \\
\hphantom{\der{}{V} L(V,V_1,V_2,\lambda_1,\lambda_2) =}{} - \lambda_2 \der{}{V} \chi \left( \frac{V+V_2}{2} ,\frac{V-V_2}{2\lambda_2},\lambda_2,\lambda_2 \right) \\
\hphantom{\der{}{V} L(V,V_1,V_2,\lambda_1,\lambda_2)}{} = \frac{\lambda_1}{2} \chi_1' \left( \frac{V+V_1}{2} ,\lambda_1,\lambda_1 \right) + \frac{1}{2} \chi_2' \left( \frac{V-V_1}{2\lambda_1},\lambda_1,\lambda_1 \right) \\
\hphantom{\der{}{V} L(V,V_1,V_2,\lambda_1,\lambda_2) =}{} - \frac{\lambda_1}{2} \chi_1' \left( \frac{V+V_2}{2},\lambda_2,\lambda_2 \right) - \frac{1}{2} \chi_2' \left( \frac{V-V_2}{2\lambda_2},\lambda_2,\lambda_2 \right) \\
\hphantom{\der{}{V} L(V,V_1,V_2,\lambda_1,\lambda_2)}{} = \psi \left( \frac{V+V_1}{2} ,\frac{V-V_1}{2\lambda_1},\lambda_1,\lambda_1 \right) - \psi \left( \frac{V+V_2}{2} ,\frac{V-V_2}{2\lambda_2},\lambda_2,\lambda_2 \right) \\
\hphantom{\der{}{V} L(V,V_1,V_2,\lambda_1,\lambda_2)}{}= \Psi\big(V,V_1,\lambda_1^2\big) - \Psi\big(V,V_2,\lambda_2^2\big) .
\end{gather*}
Similarly, using the symmetries of $\Lambda$, we have
\begin{gather*} \der{}{V_1} L(V,V_1,V_2,\lambda_1,\lambda_2)
= \Psi\big(V_1,V,\lambda_1^2\big) - \Psi\big(V_1,V_2,\lambda_1^2-\lambda_2^2\big) \end{gather*}
and
\begin{gather*} \der{}{V_2} L(V,V_1,V_2,\lambda_1,\lambda_2)
= -\Psi\big(V_2,V,\lambda_2^2\big) - \Psi\big(V_2,V_1,\lambda_1^2-\lambda_2^2\big) . \end{gather*}
Summing up all derivatives of the action in the plane with respect to the field at one vertex, and using the symmetry of the quad equation, we find two shifted copies of equation~\eqref{Q-3-leg}, arranged as in Fig.~\ref{fig-3-leg-form}$(d)$.
\end{proof}

The Lagrangian constructed this way is suitable for the continuum limit procedure, as the following proposition establishes.
\begin{Proposition}Every term of $L$ is of at least first order in both parameters, $L = \cO(\lambda_1 \lambda_2)$.
\end{Proposition}
\begin{proof}Taking the limit $\lambda_1 \rightarrow 0$, the Lagrangian vanishes:
\begin{align*}
L(V,V,V_2,0,\lambda_2) &= - \lambda_2 \chi \left( \frac{V+V_2}{2} ,\frac{V-V_2}{2\lambda_2},\lambda_2,\lambda_2 \right) + \lambda_2 \chi \left( \frac{V+V_2}{2} ,\frac{V-V_2}{-2\lambda_2},-\lambda_2,\lambda_2 \right) \\
&= 0.
\end{align*}
Similarly, for $\lambda_2 \rightarrow 0$ we have $L(V,V_1,V,\lambda_1,0) = 0$.
\end{proof}

\subsection[Q1$_{\delta=0}$]{$\boldsymbol{{\rm Q1}_{\delta=0}}$}\label{sec-Q10}

A continuum limit for equation Q1,
\begin{gather*}\lambda_1^2 (V_2 - V) (V_{12} - V_1) - \lambda_2^2 (V_1 - V) (V_{12} - V_2) = 0,\end{gather*}
with its pluri-Lagrangian structure, was given in \cite{vermeeren2019continuum}. The result is the hierarchy of Schwarzian KdV equations,
\begin{gather*}
v_{2} = 0, \nonumber\\
\frac{v_{3}}{v_{1}} = -\frac{3 v_{11}^{2}}{2 v_{1}^{2}} + \frac{v_{111}}{v_{1}}, \nonumber\\
v_{4} = 0,\nonumber \\
\frac{v_{5}}{v_{1}} = -\frac{45 v_{11}^{4}}{8 v_{1}^{4}} + \frac{25 v_{11}^{2} v_{111}}{2 v_{1}^{3}} - \frac{5 v_{111}^{2}}{2 v_{1}^{2}} - \frac{5 v_{11} v_{1111}}{v_{1}^{2}} + \frac{v_{11111}}{v_{1}},\\ 
\cdots \nonumber
\end{gather*}
where $v_i$ is a shorthand notation for the derivative $v_{t_i}$. Table \ref{table-Q10-disc} shows how this continuum limit fits in the scheme of Section~\ref{sec-Q-3-leg}. We find the discrete Lagrangian
\begin{gather*}
L = \lambda_i^2 \log \left( \frac{V - V_i}{\lambda_i} \right) - \lambda_j^2 \log \left( \frac{V - V_j}{\lambda_j} \right) - \big(\lambda_i^2 - \lambda_j^2\big) \log \left( \frac{V_i - V_j}{\lambda_i - \lambda_j} \right) .
\end{gather*}

\begin{table}[h!]
\fillbox{
\begin{minipage}{.45\linewidth}
\begin{gather*}
\Psi\big(V,W,\lambda^2\big) = \frac{\lambda^2}{V-W} \\
\psi(v,v',\lambda,\mu) = \frac{\mu}{2v'}
\end{gather*}
\end{minipage}
\begin{minipage}{.45\linewidth}
\begin{gather*}
\chi(v,v',\lambda,\mu) = \mu \log(v') \\
\Lambda(V,W,\lambda,\mu) = \lambda \mu \log \left( \frac{V-W}{2\lambda} \right)
\end{gather*}
\end{minipage}}
\caption{Q1$_{\delta=0}$ fact sheet. See Section \ref{sec-Q-3-leg} for the meaning of these functions.}\label{table-Q10-disc}
\end{table}

The first few coefficients of the continuous pluri-Lagrangian 2-form, after eliminating alien derivatives, are
\begin{gather*}
\cL_{12} = -\frac{v_{2}}{4 v_{1}},\\
\cL_{13} = \frac{v_{11}^2}{4 v_{1}^2} - \frac{v_{3}}{4 v_{1}},\\
\cL_{14} = -\frac{v_{4}}{4 v_{1}},\\
\cL_{15} = \frac{3 v_{11}^4}{16 v_{1}^4} - \frac{v_{111}^2}{4 v_{1}^2} - \frac{v_{5}}{4 v_{1}},\\
\cL_{23} = \frac{v_{11} v_{12}}{2 v_{1}^2} + \frac{v_{11}^2 v_{2}}{8 v_{1}^3} - \frac{v_{111} v_{2}}{4 v_{1}^2},\\
\cL_{24} = 0,\\
\cL_{25} = -\frac{v_{111} v_{112}}{2 v_{1}^2} + \frac{3 v_{11}^3 v_{12}}{4 v_{1}^4} - \frac{v_{11} v_{111} v_{12}}{v_{1}^3} + \frac{v_{1111} v_{12}}{2 v_{1}^2} + \frac{27 v_{11}^4 v_{2}}{32 v_{1}^5}\\
\hphantom{\cL_{25} =}{} - \frac{17 v_{11}^2 v_{111} v_{2}}{8 v_{1}^4} + \frac{7 v_{111}^2 v_{2}}{8 v_{1}^3} + \frac{3 v_{11} v_{1111} v_{2}}{4 v_{1}^3} - \frac{v_{11111} v_{2}}{4 v_{1}^2},\\
\cL_{34} = -\frac{v_{11} v_{14}}{2 v_{1}^2} - \frac{v_{11}^2 v_{4}}{8 v_{1}^3} + \frac{v_{111} v_{4}}{4 v_{1}^2},\\
\cL_{35} = \frac{45 v_{11}^6}{64 v_{1}^6} - \frac{57 v_{11}^4 v_{111}}{32 v_{1}^5} + \frac{19 v_{11}^2 v_{111}^2}{16 v_{1}^4} - \frac{7 v_{111}^3}{8 v_{1}^3} + \frac{3 v_{11}^3 v_{1111}}{8 v_{1}^4} + \frac{3 v_{11} v_{111} v_{1111}}{4 v_{1}^3} - \frac{v_{1111}^2}{4 v_{1}^2} \\
\hphantom{\cL_{35} =}{} - \frac{3 v_{11}^2 v_{11111}}{8 v_{1}^3} + \frac{v_{111} v_{11111}}{4 v_{1}^2} - \frac{v_{111} v_{113}}{2 v_{1}^2} + \frac{3 v_{11}^3 v_{13}}{4 v_{1}^4} - \frac{v_{11} v_{111} v_{13}}{v_{1}^3} + \frac{v_{1111} v_{13}}{2 v_{1}^2} - \frac{v_{11} v_{15}}{2 v_{1}^2} \\
\hphantom{\cL_{35} =}{}+ \frac{27 v_{11}^4 v_{3}}{32 v_{1}^5}- \frac{17 v_{11}^2 v_{111} v_{3}}{8 v_{1}^4} + \frac{7 v_{111}^2 v_{3}}{8 v_{1}^3} + \frac{3 v_{11} v_{1111} v_{3}}{4 v_{1}^3} - \frac{v_{11111} v_{3}}{4 v_{1}^2} - \frac{v_{11}^2 v_{5}}{8 v_{1}^3} + \frac{v_{111} v_{5}}{4 v_{1}^2},\\
\cL_{45} = -\frac{v_{111} v_{114}}{2 v_{1}^2} + \frac{3 v_{11}^3 v_{14}}{4 v_{1}^4} - \frac{v_{11} v_{111} v_{14}}{v_{1}^3} + \frac{v_{1111} v_{14}}{2 v_{1}^2} + \frac{27 v_{11}^4 v_{4}}{32 v_{1}^5} \\
\hphantom{\cL_{45} =}{}- \frac{17 v_{11}^2 v_{111} v_{4}}{8 v_{1}^4} + \frac{7 v_{111}^2 v_{4}}{8 v_{1}^3} + \frac{3 v_{11} v_{1111} v_{4}}{4 v_{1}^3} - \frac{v_{11111} v_{4}}{4 v_{1}^2}.
\end{gather*}

The even-numbered times correspond to trivial equations, $v_{2k} = 0$, restricting the dynamics to a space of half the dimension. We can also restrict the pluri-Lagrangian formulation to this space:
\begin{gather*} \cL = \sum_{i<j} \cL_{2i+1,2j+1} \d t_{2i+1} \wedge \d t_{2j+1} \end{gather*}
is a pluri-Lagrangian $2$-form for the hierarchy of nontrivial Schwarzian KdV equations
\begin{gather*}
\frac{v_{3}}{v_{1}} = -\frac{3 v_{11}^{2}}{2 v_{1}^{2}} + \frac{v_{111}}{v_{1}}, \\
\frac{v_{5}}{v_{1}} = -\frac{45 v_{11}^{4}}{8 v_{1}^{4}} + \frac{25 v_{11}^{2} v_{111}}{2 v_{1}^{3}} - \frac{5 v_{111}^{2}}{2 v_{1}^{2}} - \frac{5 v_{11} v_{1111}}{v_{1}^{2}} + \frac{v_{11111}}{v_{1}}, \\
\cdots
\end{gather*}

\subsection[Q1$_{\delta=1}$]{$\boldsymbol{{\rm Q1}_{\delta=1}}$}

Equation Q1$_{\delta=1}$ reads
\begin{gather*} \lambda_1^2 (V_2 - V) (V_{12} - V_1) - \lambda_2^2 (V_1 - V) (V_{12} - V_2) + \lambda_1^2 \lambda_2^2\big(\lambda_1^2 - \lambda_2^2\big) = 0.\end{gather*}
We apply the procedure of Section~\ref{sec-Q-3-leg} to find a Lagrangian that is suitable for the continuum limit. Intermediate steps are listed in Table \ref{table-Q11-disc}.

\begin{table}[h]
\fillbox{
\begin{gather*}
\Psi\big(V,W,\lambda^2\big) = \log \left( \frac{V-W + \lambda^2}{V-W - \lambda^2} \right) \\
\psi(v,v',\lambda,\mu) = \log \left( \frac{2 v' + \mu}{2 v' - \mu} \right) \\
\chi(v,v',\lambda,\mu) = (2 v' + \mu) \log(2 v' + \mu) - (2 v' - \mu) \log(2 v' - \mu) \\
\Lambda(V,W,\lambda,\mu) = ( V - W + \lambda \mu) \log( V-W + \lambda \mu) - (V - W - \lambda \mu) \log( V-W - \lambda \mu)
\end{gather*}}
\caption{Q1$_{\delta=1}$ fact sheet. See Section \ref{sec-Q-3-leg} for the meaning of these functions.}\label{table-Q11-disc}
\end{table}
In the continuum limit of the equation we find
\begin{gather*}
v_3 = v_{111} - \frac{3}{2} \frac{v_{11}^2 - \frac{1}{4}}{v_1}, \\
v_5 = -\frac{45 v_{11}^4}{8 v_1^3} + \frac{25 v_{11}^2 v_{111}}{2 v_1^2} - \frac{5 v_{111}^2}{2 v_1} - \frac{5 v_{11} v_{1111}}{v_1} + v_{11111} + \frac{25 v_{11}^2}{16 v_1^3} - \frac{5 v_{111}}{8 v_1^2} -
\frac{5}{128 v_1^3} , \\
\cdots
\end{gather*}
and $v_{2k} = 0$ for all $k \in \N$. Some coefficients of the continuous pluri-Lagrangian 2-form are
\begin{gather*}
\cL_{13} = \frac{v_{11}^2}{2 v_{1}^2} - \frac{v_{3}}{2 v_{1}} + \frac{1}{8 v_{1}^2},\\
\cL_{15} = \frac{3 v_{11}^4}{8 v_{1}^4} - \frac{v_{111}^2}{2 v_{1}^2} - \frac{v_{5}}{2 v_{1}} - \frac{5 v_{11}^2}{16 v_{1}^4} - \frac{1}{128 v_{1}^4},\\
\cL_{35} = \frac{45 v_{11}^6}{32 v_{1}^6} - \frac{57 v_{11}^4 v_{111}}{16 v_{1}^5} + \frac{19 v_{11}^2 v_{111}^2}{8 v_{1}^4} - \frac{7 v_{111}^3}{4 v_{1}^3} + \frac{3 v_{11}^3 v_{1111}}{4 v_{1}^4} + \frac{3 v_{11} v_{111} v_{1111}}{2 v_{1}^3} - \frac{v_{1111}^2}{2 v_{1}^2} \\
\hphantom{\cL_{35} =}{}
- \frac{3 v_{11}^2 v_{11111}}{4 v_{1}^3} + \frac{v_{111} v_{11111}}{2 v_{1}^2} - \frac{v_{111} v_{113}}{v_{1}^2} + \frac{3 v_{11}^3 v_{13}}{2 v_{1}^4} - \frac{2 v_{11} v_{111} v_{13}}{v_{1}^3} + \frac{v_{1111} v_{13}}{v_{1}^2} \\
\hphantom{\cL_{35} =}{}
- \frac{v_{11} v_{15}}{v_{1}^2} + \frac{27 v_{11}^4 v_{3}}{16 v_{1}^5} - \frac{17 v_{11}^2 v_{111} v_{3}}{4 v_{1}^4} + \frac{7 v_{111}^2 v_{3}}{4 v_{1}^3} + \frac{3 v_{11} v_{1111} v_{3}}{2 v_{1}^3} - \frac{v_{11111} v_{3}}{2 v_{1}^2} - \frac{v_{11}^2 v_{5}}{4 v_{1}^3} \\
\hphantom{\cL_{35} =}{}
+ \frac{v_{111} v_{5}}{2 v_{1}^2} - \frac{95 v_{11}^4}{128 v_{1}^6} + \frac{41 v_{11}^2 v_{111}}{32 v_{1}^5} - \frac{27 v_{111}^2}{32 v_{1}^4} - \frac{3 v_{11} v_{1111}}{16 v_{1}^4} + \frac{3 v_{11111}}{16 v_{1}^3} - \frac{5 v_{11} v_{13}}{8 v_{1}^4} - \frac{15 v_{11}^2 v_{3}}{32 v_{1}^5}\\
\hphantom{\cL_{35} =}{}
 + \frac{5 v_{111} v_{3}}{16 v_{1}^4} + \frac{v_{5}}{16 v_{1}^3} + \frac{55 v_{11}^2}{512 v_{1}^6} - \frac{25 v_{111}}{256 v_{1}^5} + \frac{3 v_{3}}{256 v_{1}^5} - \frac{5}{2048 v_{1}^6}.
\end{gather*}

\subsection{Q2}

For the equation
\begin{gather*}
 \lambda_1^2 \big( V_2^2 - V^2 \big) \big( V_{12}^2 - V_1^2 \big) - \lambda_2^2 \big( V_1^2 - V^2 \big) \big( V_{12}^2 - V_2^2 \big) \\
\qquad{} + \lambda_1^2 \lambda_2^2 \big( \lambda_1^2 - \lambda_2^2 \big) \big( V^2 + V_1^2 + V_2^2 + V_{12}^2 - \lambda_1^4 + \lambda_1^2 \lambda_2^2 - \lambda_2^4 \big) = 0
\end{gather*}
the general strategy of Section \ref{sec-Q-3-leg} works with the choices listed in Table \ref{table-Q2-disc}.

\begin{table}[h]
\fillbox{
\begin{gather*}
\Psi(V,W,\lambda^2) = \log \left( \frac{\big(V + W + \lambda^2\big) \big(V - W + \lambda^2\big)}{\big( V + W - \lambda^2\big) \big(V - W - \lambda^2\big)} \right) \\
\psi(v,v',\lambda,\mu) = \log \left( \frac{(2v + \lambda \mu) (2v' + \mu)}{( 2v - \lambda \mu) (2v' - \mu)} \right)\\
\chi(v,v',\lambda,\mu) = \frac{1}{\lambda} (2 v + \lambda \mu) \log(2 v + \lambda \mu) + (2 v' + \mu) \log(2 v' + \mu) \\
\hphantom{\chi(v,v',\lambda,\mu) =}{} - \frac{1}{\lambda}(2 v - \lambda \mu) \log(2 v - \lambda \mu) - (2 v' - \mu) \log(2 v' - \mu) \\
 \Lambda(V,W,\lambda,\mu) = (V + W + \lambda \mu) \log( V + W + \lambda \mu) + (V - W + \lambda \mu) \log \left( \frac{V - W}{\lambda} + \mu \right) \\
\hphantom{\Lambda(V,W,\lambda,\mu) =}{} - (V + W - \lambda \mu) \log( V + W - \lambda \mu) - (V - W - \lambda \mu) \log \left( \frac{V - W}{\lambda} - \mu \right)
\end{gather*}
}
\caption{Q2 fact sheet. See Section \ref{sec-Q-3-leg} for the meaning of these functions.}\label{table-Q2-disc}
\end{table}

 The continuum limit hierarchy is
\begin{gather*}
v_3 = v_{111} - \frac{3}{2} \frac{v_{11}^2 - \frac{1}{4}}{v_1} - \frac{3}{2} \frac{v_1^3}{v^2}, \\
v_{5} = -\frac{45 v_{1}^{5}}{8 v^{4}} + \frac{15 v_{1}^{3} v_{11}}{v^{3}} + \frac{15 v_{1} v_{11}^{2}}{4 v^{2}} - \frac{45v_{11}^{4}}{8 v_{1}^{3}} - \frac{15 v_{1}^{2} v_{111}}{2 v^{2}}
+ \frac{25 v_{11}^{2} v_{111}}{2 v_{1}^{2}} - \frac{5v_{111}^{2}}{2 v_{1}} \\
\hphantom{v_{5} =}{} - \frac{5 v_{11} v_{1111}}{v_{1}} + v_{11111} + \frac{5 v_{1}}{16 v^{2}} + \frac{25 v_{11}^{2}}{16v_{1}^{3}} - \frac{5 v_{111}}{8 v_{1}^{2}} - \frac{5}{128v_{1}^{3}} , \\
\cdots
\end{gather*}
and $v_{2k} = 0$ for all $k \in \N$. A few coefficients of the pluri-Lagrangian 2-form are
\begin{gather*}
\cL_{13} = \frac{3 v_{1}^2}{2 v_{}^2} + \frac{v_{11}^2}{2 v_{1}^2} - \frac{v_{3}}{2 v_{1}} + \frac{1}{8 v_{1}^2},\\
\cL_{15} = \frac{15 v_{1}^4}{8 v_{}^4} - \frac{15 v_{11}^2}{4 v_{}^2} + \frac{3 v_{11}^4}{8 v_{1}^4} - \frac{v_{111}^2}{2 v_{1}^2} - \frac{v_{5}}{2 v_{1}} + \frac{5}{16 v_{}^2} - \frac{5 v_{11}^2}{16 v_{1}^4} - \frac{1}{128 v_{1}^4},\\
\cL_{35} = \frac{45 v_{1}^6}{32 v_{}^6} - \frac{9 v_{1}^4 v_{11}}{4 v_{}^5} + \frac{63 v_{1}^2 v_{11}^2}{32 v_{}^4} - \frac{81 v_{11}^3}{4 v_{}^3} + \frac{495 v_{11}^4}{32 v_{}^2 v_{1}^2} + \frac{45 v_{11}^6}{32 v_{1}^6} - \frac{9 v_{1}^3 v_{111}}{16 v_{}^4} + \frac{39 v_{1} v_{11} v_{111}}{2 v_{}^3} \\
\hphantom{\cL_{35} =}{}
- \frac{165 v_{11}^2 v_{111}}{8 v_{}^2 v_{1}} - \frac{57 v_{11}^4 v_{111}}{16 v_{1}^5} + \frac{3 v_{111}^2}{8 v_{}^2} + \frac{19 v_{11}^2 v_{111}^2}{8 v_{1}^4} - \frac{7 v_{111}^3}{4 v_{1}^3} - \frac{3 v_{1}^2 v_{1111}}{v_{}^3} + \frac{27 v_{11} v_{1111}}{4 v_{}^2} \\
\hphantom{\cL_{35} =}{}
+ \frac{3 v_{11}^3 v_{1111}}{4 v_{1}^4} + \frac{3 v_{11} v_{111} v_{1111}}{2 v_{1}^3} - \frac{v_{1111}^2}{2 v_{1}^2} - \frac{3 v_{1} v_{11111}}{4 v_{}^2} - \frac{3 v_{11}^2 v_{11111}}{4 v_{1}^3} + \frac{v_{111} v_{11111}}{2 v_{1}^2} \\
\hphantom{\cL_{35} =}{}
- \frac{v_{111} v_{113}}{v_{1}^2} - \frac{15 v_{11} v_{13}}{2 v_{}^2} + \frac{3 v_{11}^3 v_{13}}{2 v_{1}^4} - \frac{2 v_{11} v_{111} v_{13}}{v_{1}^3} + \frac{v_{1111} v_{13}}{v_{1}^2} - \frac{v_{11} v_{15}}{v_{1}^2} + \frac{75 v_{1}^3 v_{3}}{16 v_{}^4} \\
\hphantom{\cL_{35} =}{}
- \frac{15 v_{1} v_{11} v_{3}}{2 v_{}^3} + \frac{15 v_{11}^2 v_{3}}{8 v_{}^2 v_{1}} + \frac{27 v_{11}^4 v_{3}}{16 v_{1}^5} + \frac{15 v_{111} v_{3}}{4 v_{}^2} - \frac{17 v_{11}^2 v_{111} v_{3}}{4 v_{1}^4} + \frac{7 v_{111}^2 v_{3}}{4 v_{1}^3} \\
\hphantom{\cL_{35} =}{}
+ \frac{3 v_{11} v_{1111} v_{3}}{2 v_{1}^3} - \frac{v_{11111} v_{3}}{2 v_{1}^2} - \frac{9 v_{1} v_{5}}{4 v_{}^2} - \frac{v_{11}^2 v_{5}}{4 v_{1}^3} + \frac{v_{111} v_{5}}{2 v_{1}^2} - \frac{135 v_{1}^2}{128 v_{}^4} + \frac{65 v_{11}}{16 v_{}^3} - \frac{255 v_{11}^2}{64 v_{}^2 v_{1}^2} \\
\hphantom{\cL_{35} =}{}
- \frac{95 v_{11}^4}{128 v_{1}^6} - \frac{35 v_{111}}{32 v_{}^2 v_{1}} + \frac{41 v_{11}^2 v_{111}}{32 v_{1}^5} - \frac{27 v_{111}^2}{32 v_{1}^4} - \frac{3 v_{11} v_{1111}}{16 v_{1}^4} + \frac{3 v_{11111}}{16 v_{1}^3} - \frac{5 v_{11} v_{13}}{8 v_{1}^4} + \frac{5 v_{3}}{32 v_{}^2 v_{1}} \\
\hphantom{\cL_{35} =}{}
- \frac{15 v_{11}^2 v_{3}}{32 v_{1}^5} + \frac{5 v_{111} v_{3}}{16 v_{1}^4} + \frac{v_{5}}{16 v_{1}^3} + \frac{15}{512 v_{}^2 v_{1}^2} + \frac{55 v_{11}^2}{512 v_{1}^6} - \frac{25 v_{111}}{256 v_{1}^5} + \frac{3 v_{3}}{256 v_{1}^5} - \frac{5}{2048 v_{1}^6}.
\end{gather*}

\subsection[Q3$_{\delta=0}$]{$\boldsymbol{{\rm Q3}_{\delta=0}}$}\label{sec-Q30}

We consider equation Q3 in the form
\begin{gather*}
\sin\big(\lambda_1^2\big) \big( {\rm e}^{iV} {\rm e}^{iV_1} + {\rm e}^{iV_2} {\rm e}^{iV_{12}}\big) -\sin\big(\lambda_2^2\big) \big( {\rm e}^{iV} {\rm e}^{iV_2} + {\rm e}^{iV_1} {\rm e}^{iV_{12}}\big) \\
\qquad{} - \sin\big(\lambda_1^2 - \lambda_2^2\big)\big( {\rm e}^{iV} {\rm e}^{iV_{12}} + {\rm e}^{iV_1} {\rm e}^{iV_2}\big) = 0 .
\end{gather*}
The strategy of Section \ref{sec-Q-3-leg} works as outlined in Table \ref{table-Q30-disc}.

\begin{table}[h]
\fillbox{
\begin{gather*}
\Psi\big(V,W,\lambda^2\big) = \log \left( \frac{{\rm e}^{i \lambda^2} {\rm e}^{i V} - {\rm e}^{i W}}{ {\rm e}^{i V} - {\rm e}^{i \lambda^2} {\rm e}^{i W}} \right) \\
\hphantom{\Psi\big(V,W,\lambda^2\big)}{} = \log \left( \sin \left( \frac{V-W+\lambda^2}{2} \right) \right) - \log \left( \sin \left( \frac{V-W-\lambda^2}{2} \right) \right) \\
\psi(v,v',\lambda,\mu) = \log \left( \sin \left( \lambda v' + \frac{\lambda \mu}{2} \right) \right) - \log \left( \sin \left(\lambda v' - \frac{\lambda \mu}{2} \right) \right) \\
\chi(v,v',\lambda,\mu) = \frac{i}{\lambda} \big( {-} 2 \lambda^2 \mu v' + \Li2 \big( {\rm e}^{i(2 \lambda v' + \lambda \mu)} \big) - \Li2 \big( {\rm e}^{i(2 \lambda v' - \lambda \mu)}\big)\big)
\\
 \Lambda(V,W,\lambda,\mu) = \lambda \mu (V - W) - \Li2 \big({\rm e}^{i(V-W+\lambda \mu)} \big) + \Li2 \big({\rm e}^{i(V-W-\lambda \mu)} \big)
\end{gather*}}
\caption{Q3$_{\delta=0}$ fact sheet. See Section \ref{sec-Q-3-leg} for the meaning of these functions.}\label{table-Q30-disc}
\end{table}

The continuum limit hierarchy is
\begin{gather*}
v_3 = v_{111} - \frac{3}{2} \frac{v_{11}^2 - \frac{1}{4}}{v_1} + \frac{1}{2} v_1^3, \\
v_{5} = \frac{3}{8} v_{1}^{5} - \frac{5}{4} v_{1} v_{11}^{2} + \frac{5}{2} v_{1}^{2} v_{111} - \frac{5}{48} v_{1} - \frac{45 v_{11}^{4}}{8 v_{1}^{3}} + \frac{25 v_{11}^{2} v_{111}}{2 v_{1}^{2}} - \frac{5 v_{111}^{2}}{2 v_{1}} - \frac{5 v_{11} v_{1111}}{v_{1}} \\
\hphantom{v_{5} =}{} + v_{11111} + \frac{25 v_{11}^{2}}{16 v_{1}^{3}} - \frac{5 v_{111}}{8 v_{1}^{2}} - \frac{5}{128 v_{1}^{3}}, \\
\cdots
\end{gather*}
and $v_{2k} = 0$ for all $k \in \N$. Some coefficients of the continuous pluri-Lagrangian 2-form are
\begin{gather*}
\cL_{13} = -\frac{1}{4} v_{1}^2 + \frac{v_{11}^2}{4 v_{1}^2} - \frac{v_{3}}{4 v_{1}} + \frac{1}{16 v_{1}^2},\\
\cL_{15} = -\frac{1}{16} v_{1}^4 + \frac{5}{8} v_{11}^2 + \frac{3 v_{11}^4}{16 v_{1}^4} - \frac{v_{111}^2}{4 v_{1}^2} - \frac{v_{5}}{4 v_{1}} - \frac{5 v_{11}^2}{32 v_{1}^4} - \frac{1}{256 v_{1}^4} - \frac{5}{96},\\
\cL_{35} = \frac{1}{64} v_{1}^6 - \frac{17}{64} v_{1}^2 v_{11}^2 + \frac{7}{32} v_{1}^3 v_{111} - \frac{5}{32} v_{1}^3 v_{3} + \frac{91}{768} v_{1}^2 - \frac{165 v_{11}^4}{64 v_{1}^2} + \frac{55 v_{11}^2 v_{111}}{16 v_{1}} - \frac{1}{16} v_{111}^2 \\
\hphantom{\cL_{35} =}{}
- \frac{9}{8} v_{11} v_{1111} + \frac{1}{8} v_{1} v_{11111} + \frac{5}{4} v_{11} v_{13} - \frac{5 v_{11}^2 v_{3}}{16 v_{1}} - \frac{5}{8} v_{111} v_{3} + \frac{3}{8} v_{1} v_{5} + \frac{85 v_{11}^2}{128 v_{1}^2} + \frac{45 v_{11}^6}{64 v_{1}^6} \\
\hphantom{\cL_{35} =}{}
+ \frac{35 v_{111}}{192 v_{1}} - \frac{57 v_{11}^4 v_{111}}{32 v_{1}^5} + \frac{19 v_{11}^2 v_{111}^2}{16 v_{1}^4} - \frac{7 v_{111}^3}{8 v_{1}^3} + \frac{3 v_{11}^3 v_{1111}}{8 v_{1}^4} + \frac{3 v_{11} v_{111} v_{1111}}{4 v_{1}^3} - \frac{v_{1111}^2}{4 v_{1}^2} \\
\hphantom{\cL_{35} =}{}
- \frac{3 v_{11}^2 v_{11111}}{8 v_{1}^3} + \frac{v_{111} v_{11111}}{4 v_{1}^2} - \frac{v_{111} v_{113}}{2 v_{1}^2} + \frac{3 v_{11}^3 v_{13}}{4 v_{1}^4} - \frac{v_{11} v_{111} v_{13}}{v_{1}^3} + \frac{v_{1111} v_{13}}{2 v_{1}^2} - \frac{v_{11} v_{15}}{2 v_{1}^2} \\
\hphantom{\cL_{35} =}{}
- \frac{5 v_{3}}{192 v_{1}} + \frac{27 v_{11}^4 v_{3}}{32 v_{1}^5} - \frac{17 v_{11}^2 v_{111} v_{3}}{8 v_{1}^4} + \frac{7 v_{111}^2 v_{3}}{8 v_{1}^3} + \frac{3 v_{11} v_{1111} v_{3}}{4 v_{1}^3} - \frac{v_{11111} v_{3}}{4 v_{1}^2} - \frac{v_{11}^2 v_{5}}{8 v_{1}^3} \\
\hphantom{\cL_{35} =}{}
+ \frac{v_{111} v_{5}}{4 v_{1}^2} - \frac{5}{1024 v_{1}^2} - \frac{95 v_{11}^4}{256 v_{1}^6} + \frac{41 v_{11}^2 v_{111}}{64 v_{1}^5} - \frac{27 v_{111}^2}{64 v_{1}^4} - \frac{3 v_{11} v_{1111}}{32 v_{1}^4} + \frac{3 v_{11111}}{32 v_{1}^3} \\
\hphantom{\cL_{35} =}{}
 - \frac{5 v_{11} v_{13}}{16 v_{1}^4}- \frac{15 v_{11}^2 v_{3}}{64 v_{1}^5} + \frac{5 v_{111} v_{3}}{32 v_{1}^4} + \frac{v_{5}}{32 v_{1}^3} + \frac{55 v_{11}^2}{1024 v_{1}^6} - \frac{25 v_{111}}{512 v_{1}^5} + \frac{3 v_{3}}{512 v_{1}^5} - \frac{5}{4096 v_{1}^6}.
\end{gather*}

{\bf Note on dilogarithms.} In Table~\ref{table-Q30-disc} we encounter the dilogarithm function~\cite{zagier2007dilogarithm}, given by
\begin{gather*} \Li2(z) = \sum_{n = 1}^\infty \frac{z^n}{n^2} \qquad \text{for } |z| < 1 \end{gather*}
and by analytic continuation for other $z \in \C \setminus [1,\infty)$. Its derivative is
\begin{gather*} \frac{\d}{\d z} \Li2(z)= - \frac{\log(1 - z)}{z} . \end{gather*}
This can be used to integrate $\log \circ \sin$, which occurs in the leg function $\Psi$,
\begin{gather*} \int \log(\sin z) \,\d z = \frac{i}{2} \big( {-} z^2 + \Li2 \big( {\rm e}^{2iz} \big) \big) - z \log(-2 i) + c. \end{gather*}
In the continuum limit procedure we need a series expansion of $\Li2({\rm e}^x)$. For small $x$ this seems problematic, since $\Li2(z)$ has a branch point at $z = 1$. The way to handle this is to carry along this branch point in the formal series expansion
\begin{gather*}
\Li2({\rm e}^x) = \Li2(1) - x \log(-x) + x - \frac{x^2}{4} - \frac{x^3}{72} + \cdots,
\end{gather*}
where the first term famously equals $\Li2(1) = \frac{\pi^2}{6}$.

\subsection[Q3$_{\delta=1}$]{$\boldsymbol{{\rm Q3}_{\delta=1}}$}

For equation Q3 with nonzero parameter we use a slightly different form compared to Q3$_{\delta=0}$,
\begin{gather*}
 \sin\big(\lambda_1^2\big) \big( \cos(V) \cos(V_1) + \cos(V_2) \cos(V_{12}) \big) \\
\qquad{}- \sin\big(\lambda_2^2\big) \big( \cos(V) \cos(V_2) + \cos(V_1) \cos(V_{12}) \big) \\
\qquad{} - \sin\big(\lambda_1^2 - \lambda_2^2\big) \big( \cos(V) \cos(V_{12}) + \cos(V_1) \cos(V_2) - \sin\big(\lambda_1^2\big) \sin\big(\lambda_2^2\big)\big) = 0.
\end{gather*}
The general strategy of Section \ref{sec-Q-3-leg} applies to this transformed equation with the choices in Table \ref{table-Q31-disc}.

\begin{table}[h]
\fillbox{
\begin{gather*}
\Psi\big(V,W,\lambda^2\big) = \log \left( \sin \left( \frac{V-W+\lambda^2}{2} \right) \sin \left( \frac{V+W+\lambda^2}{2} \right) \right) \\
\hphantom{\Psi\big(V,W,\lambda^2\big) =}{} - \log \left( \sin \left( \frac{V-W-\lambda^2}{2} \right) \sin \left( \frac{V+W-\lambda^2}{2} \right) \right)\\
\psi(v,v',\lambda,\mu) = \log \left( \sin \left( \lambda v' + \frac{\lambda \mu}{2} \right) \sin \left( v + \frac{\lambda \mu}{2} \right) \right) \\
\hphantom{\psi(v,v',\lambda,\mu) =}{} - \log \left( \sin \left( \lambda v' - \frac{\lambda \mu}{2} \right) \sin \left( v - \frac{\lambda \mu}{2} \right) \right)\\
\chi(v,v',\lambda,\mu) = \frac{i}{\lambda} \big({-} 2 \lambda^2 \mu v' + \Li2 \big( {\rm e}^{i(2 \lambda v' + \lambda \mu)} \big) - \Li2 \big( {\rm e}^{i(2 \lambda v' - \lambda \mu)} \big)\big) \\
\hphantom{\chi(v,v',\lambda,\mu) =}{} - \frac{i}{\lambda} \big( {-} 2 \lambda \mu v + \Li2 \big( {\rm e}^{i(2 v + \lambda \mu)} \big) - \Li2 \big( {\rm e}^{i(2 v - \lambda \mu)} \big)\big) \\
\Lambda(V,W,\lambda,\mu) = \lambda \mu (V - W) - \Li2\big({\rm e}^{i(V-W+\lambda \mu)}\big) + \Li2\big({\rm e}^{i(V-W-\lambda \mu)} \big) \\
\hphantom{\Lambda(V,W,\lambda,\mu) =}{} + \lambda \mu (V + W) - \Li2 \big({\rm e}^{i(V+W+\lambda\mu)} \big) + \Li2 \big({\rm e}^{i(V+W-\lambda\mu)}\big)
\end{gather*}}
\caption{Q3$_{\delta=1}$ fact sheet. See Section \ref{sec-Q-3-leg} for the meaning of these functions.}\label{table-Q31-disc}
\end{table}

The continuum limit hierarchy is
\begin{gather*}
v_3 = v_{111} - \frac{3}{2} \frac{v_{11}^2 - \frac{1}{4}}{v_1} + \frac{1}{2} v_1^3 - \frac{3}{2} \frac{v_1^3}{\sin(v)^2}, \\
v_{5} = \frac{3}{8} v_{1}^{5} - \frac{5}{4} v_{1} v_{11}^{2} + \frac{5}{2} v_{1}^{2} v_{111} + \frac{15 v_{1}^{5}}{4 \sin(v)^{2}} + \frac{15 v_{1}^{3} v_{11} \cos(v)}{\sin(v)^{3}} - \frac{5}{48} v_{1} - \frac{45 v_{11}^{4}}{8 v_{1}^{3}} + \frac{25 v_{11}^{2} v_{111}}{2 v_{1}^{2}} \\
\hphantom{v_{5} =}{}- \frac{5 v_{111}^{2}}{2 v_{1}} - \frac{5 v_{11} v_{1111}}{v_{1}} + v_{11111} - \frac{45 v_{1}^{5}}{8 \sin(v)^{4}} + \frac{15 v_{1} v_{11}^{2}}{4 \sin(v)^{2}} - \frac{15 v_{1}^{2} v_{111}}{2 \sin(v)^{2}} + \frac{25 v_{11}^{2}}{16 v_{1}^{3}} \\
\hphantom{v_{5} =}{} - \frac{5 v_{111}}{8 v_{1}^{2}} + \frac{5 v_{1}}{16 \sin (v )^{2}} - \frac{5}{128 v_{1}^{3}},\\
\cdots
\end{gather*}
and $v_{2k} = 0$ for all $k \in \N$. A few coefficients of the continuous pluri-Lagrangian 2-form are
\begin{gather*}
\cL_{13} = -\frac{1}{4} v_{1}^2 + \frac{v_{11}^2}{4 v_{1}^2} - \frac{v_{3}}{4 v_{1}} + \frac{3 v_{1}^2}{4 \sin(v)^2} + \frac{1}{16 v_{1}^2},\\
\cL_{15} = -\frac{1}{16} v_{1}^4 + \frac{5}{8} v_{11}^2 - \frac{5}{12} v_{} v_{13} - \frac{5}{12} v_{1} v_{3} - \frac{5 v_{1}^4}{8 \sin(v)^2} + \frac{3 v_{11}^4}{16 v_{1}^4} - \frac{v_{111}^2}{4 v_{1}^2} - \frac{v_{5}}{4 v_{1}} + \frac{15 v_{1}^4}{16 \sin(v)^4} \\
\hphantom{\cL_{15} =}{}
- \frac{15 v_{11}^2}{8 \sin(v)^2} - \frac{5 v_{11}^2}{32 v_{1}^4} + \frac{5}{32 \sin(v)^2} - \frac{1}{256 v_{1}^4} - \frac{5}{96},\\
\cL_{35} = \frac{23}{192} v_{1}^6 - \frac{57}{64} v_{1}^2 v_{11}^2 + \frac{61}{96} v_{1}^3 v_{111} - \frac{55}{96} v_{1}^3 v_{3} - \frac{v_{1}^6}{64 \sin(v)^2} + \frac{3 v_{1}^4 v_{11} \cos(v)}{8 \sin(v)^3} + \frac{211}{768} v_{1}^2 \\
\hphantom{\cL_{35} =}{}
- \frac{105 v_{11}^4}{64 v_{1}^2} + \frac{35 v_{11}^2 v_{111}}{16 v_{1}} + \frac{17}{48} v_{111}^2 - \frac{9}{8} v_{11} v_{1111} + \frac{1}{8} v_{1} v_{11111} + \frac{5}{4} v_{11} v_{13} + \frac{15 v_{11}^2 v_{3}}{16 v_{1}}\\
\hphantom{\cL_{35} =}{}
- \frac{35}{24} v_{111} v_{3} - \frac{5}{12} v_{} v_{33} + \frac{3}{8} v_{1} v_{5} + \frac{15 v_{1}^6}{64 \sin(v)^4} + \frac{39 v_{1}^2 v_{11}^2}{32 \sin(v)^2} - \frac{17 v_{1}^3 v_{111}}{16 \sin(v)^2} - \frac{5 v_{1}^3 v_{3}}{16 \sin(v)^2} \\
\hphantom{\cL_{35} =}{}
- \frac{9 v_{1}^4 v_{11} \cos(v)}{8 \sin(v)^5} - \frac{81 v_{11}^3 \cos(v)}{8 \sin(v)^3} + \frac{39 v_{1} v_{11} v_{111} \cos(v)}{4 \sin(v)^3} - \frac{3 v_{1}^2 v_{1111} \cos(v)}{2 \sin(v)^3} \\
\hphantom{\cL_{35} =}{}
- \frac{15 v_{1} v_{11} v_{3} \cos(v)}{4 \sin(v)^3} + \frac{25 v_{11}^2}{128 v_{1}^2} + \frac{45 v_{11}^6}{64 v_{1}^6} + \frac{95 v_{111}}{192 v_{1}} - \frac{57 v_{11}^4 v_{111}}{32 v_{1}^5} + \frac{19 v_{11}^2 v_{111}^2}{16 v_{1}^4} - \frac{7 v_{111}^3}{8 v_{1}^3} \\
\hphantom{\cL_{35} =}{}
+ \frac{3 v_{11}^3 v_{1111}}{8 v_{1}^4} + \frac{3 v_{11} v_{111} v_{1111}}{4 v_{1}^3} - \frac{v_{1111}^2}{4 v_{1}^2} - \frac{3 v_{11}^2 v_{11111}}{8 v_{1}^3} + \frac{v_{111} v_{11111}}{4 v_{1}^2} - \frac{v_{111} v_{113}}{2 v_{1}^2} + \frac{3 v_{11}^3 v_{13}}{4 v_{1}^4}\! \\
\hphantom{\cL_{35} =}{}
- \frac{v_{11} v_{111} v_{13}}{v_{1}^3} + \frac{v_{1111} v_{13}}{2 v_{1}^2} - \frac{v_{11} v_{15}}{2 v_{1}^2} - \frac{65 v_{3}}{192 v_{1}} + \frac{27 v_{11}^4 v_{3}}{32 v_{1}^5} - \frac{17 v_{11}^2 v_{111} v_{3}}{8 v_{1}^4} + \frac{7 v_{111}^2 v_{3}}{8 v_{1}^3} \\
\hphantom{\cL_{35} =}{}
+ \frac{3 v_{11} v_{1111} v_{3}}{4 v_{1}^3} - \frac{v_{11111} v_{3}}{4 v_{1}^2} - \frac{v_{11}^2 v_{5}}{8 v_{1}^3} + \frac{v_{111} v_{5}}{4 v_{1}^2} + \frac{45 v_{1}^6}{64 \sin(v)^6} + \frac{63 v_{1}^2 v_{11}^2}{64 \sin(v)^4} - \frac{9 v_{1}^3 v_{111}}{32 \sin(v)^4} \\
\hphantom{\cL_{35} =}{}
+ \frac{75 v_{1}^3 v_{3}}{32 \sin(v)^4} - \frac{15 v_{1}^2}{128 \sin(v)^2} + \frac{495 v_{11}^4}{64 v_{1}^2 \sin(v)^2} - \frac{165 v_{11}^2 v_{111}}{16 v_{1} \sin(v)^2} + \frac{3 v_{111}^2}{16 \sin(v)^2} + \frac{27 v_{11} v_{1111}}{8 \sin(v)^2} \\
\hphantom{\cL_{35} =}{}
- \frac{3 v_{1} v_{11111}}{8 \sin(v)^2} - \frac{15 v_{11} v_{13}}{4 \sin(v)^2} + \frac{15 v_{11}^2 v_{3}}{16 v_{1} \sin(v)^2} + \frac{15 v_{111} v_{3}}{8 \sin(v)^2} - \frac{9 v_{1} v_{5}}{8 \sin(v)^2} + \frac{65 v_{11} \cos(v)}{32 \sin(v)^3} \\
\hphantom{\cL_{35} =}{}
+ \frac{55}{1024 v_{1}^2} - \frac{95 v_{11}^4}{256 v_{1}^6} + \frac{41 v_{11}^2 v_{111}}{64 v_{1}^5} - \frac{27 v_{111}^2}{64 v_{1}^4} - \frac{3 v_{11} v_{1111}}{32 v_{1}^4} + \frac{3 v_{11111}}{32 v_{1}^3} - \frac{5 v_{11} v_{13}}{16 v_{1}^4} \\
\hphantom{\cL_{35} =}{}
- \frac{15 v_{11}^2 v_{3}}{64 v_{1}^5} + \frac{5 v_{111} v_{3}}{32 v_{1}^4} + \frac{v_{5}}{32 v_{1}^3} - \frac{135 v_{1}^2}{256 \sin(v)^4} - \frac{255 v_{11}^2}{128 v_{1}^2 \sin(v)^2} - \frac{35 v_{111}}{64 v_{1} \sin(v)^2}\\
\hphantom{\cL_{35} =}{}
 + \frac{5 v_{3}}{64 v_{1} \sin(v)^2} + \frac{55 v_{11}^2}{1024 v_{1}^6} - \frac{25 v_{111}}{512 v_{1}^5} + \frac{3 v_{3}}{512 v_{1}^5} + \frac{15}{1024 v_{1}^2 \sin(v)^2} - \frac{5}{4096 v_{1}^6}.
\end{gather*}

\subsection{Q4}\label{sec-Q4}

Consider equation Q4 in its elliptic form
\begin{gather*}
 A \big( (\wp(V)-b)(\wp(V_2)-b) - (a-b)(c-b) \big) \big( (\wp(V_1)-b)(\wp(V_{12})-b) - (a-b)(c-b) \big) \\
{} +B \big( (\wp(V)-a)(\wp(V_1)-a) - (b-a)(c-a) \big) \big( (\wp(V_2)-a)(\wp(V_{12})-a) - (b-a)(c-a) \big) \\
{}= ABC(a-b),
\end{gather*}
where
\begin{gather*}
 (a,A) = \big( \wp\big(\lambda_1^2\big), \wp'\big(\lambda_1^2\big) \big), \qquad
 (b,B) = \big( \wp\big(\lambda_2^2\big), \wp'\big(\lambda_2^2\big) \big), \\
 (c,C) = \big( \wp(\lambda_2^2-\lambda_1^2), \wp'(\lambda_2^2-\lambda_1^2) \big),
\end{gather*}
and $\wp$ is the Weierstrass elliptic function. The approach of Section \ref{sec-Q-3-leg} works as outlined in Table~\ref{table-Q4-disc}. A few additional functions and constants appear in this table: the Weierstrass func\-tions~$\sigma$ and~$\zeta$, with the same periods as $\wp$, and the invariants $g_2$ and $g_3$ of the $\wp$-function. They satisfy
\begin{alignat*}{3}
&\zeta= \frac{\sigma'}{\sigma}, \qquad && g_2 = 12 \wp^2 - 2 \wp'',& \\
& \zeta'= - \wp,\qquad && g_3= -8 \wp^3 + 2 \wp'' \wp - (\wp')^2.&
\end{alignat*}
Additionally, for the series expansion of $\psi$ we use
\begin{gather*} \sigma(z) = \sum_{m,n = 0}^\infty a_{m,n} \left(\frac{g_2}{2}\right)^m (2 g_3)^n \frac{z^{4m+6n+1}}{(4m+6n+1)!}, \end{gather*}
where the first coefficients are $a_{0,0} = 1$, $a_{1,0} = -1$, and $a_{0,1} = -3$. These and many more identities involving the Weierstrass elliptic functions can be found in \cite[Chapter~18]{abramowitz1964handbook}.

\begin{table}[h]
\fillbox{
\begin{gather*}
\Psi\big(V,W,\lambda^2\big) = \log \left( \frac{ \sigma\big(V+W+\lambda^2\big) \sigma\big(V-W+\lambda^2\big) }{ \sigma\big(V+W-\lambda^2\big) \sigma\big(V-W-\lambda^2\big) } \right)\\
\psi(v,v',\lambda,\mu) = \log \left( \frac{ \sigma(2 v + \lambda \mu) \sigma(2 \lambda v' + \lambda \mu) }{ \sigma(2 v - \lambda \mu) \sigma(2 \lambda v' - \lambda \mu) } \right) \\
\hphantom{\psi(v,v',\lambda,\mu)}{} = \left( 2 \zeta(2v) \lambda \mu - \frac{1}{3} \wp'(2v) \lambda^3 \mu^3 + \frac{1}{5} (\wp(2v) \wp'(2v)) \lambda^5 \mu^5 + \cdots \right) \\
\hphantom{\psi(v,v',\lambda,\mu)=}{} + \left( \log(2 v' + \mu) - \frac{g_2 (2 \lambda v' + \lambda \mu)^4}{240}-\frac{g_3 (2 \lambda v' + \lambda \mu)^6}{840} + \cdots \right) \\
\hphantom{\psi(v,v',\lambda,\mu)=}{} - \left( \log(2 v' - \mu) - \frac{g_2 (2 \lambda v' - \lambda \mu)^4}{240}-\frac{g_3 (2 \lambda v' - \lambda \mu)^6}{840} + \cdots \right)\\
\chi(v,v',\lambda,\mu) = \frac{1}{\lambda} \left( 2 \log(\sigma(2v)) \lambda \mu - \frac{1}{3} \wp(2v) \lambda^3 \mu^3 - \frac{1}{10} \wp(2v)^2 \lambda^5 \mu^5 + \cdots \right) \\
\hphantom{\chi(v,v',\lambda,\mu) =}{} + \left( (2 v' + \mu) (\log (2 v' + \mu) - 1)- \frac{g_2 (2 \lambda v' + \lambda \mu)^5}{1200 \lambda} + \cdots \right) \\
\hphantom{\chi(v,v',\lambda,\mu) =}{}- \left( (2 v' - \mu) (\log (2 v' - \mu) - 1) - \frac{g_2 (2 \lambda v' - \lambda \mu)^5}{1200 \lambda} + \cdots \right)\\
\Lambda(V,W,\lambda,\mu) = 2 \log(\sigma(V+W)) \lambda \mu - \frac{1}{3} \wp(V+W) \lambda^3 \mu^3 - \frac{1}{10} \wp(V+W)^2 \lambda^5 \mu^5 + \cdots \\
\hphantom{\Lambda(V,W,\lambda,\mu) =}{} + (V - W + \mu \lambda) \left(\log \left(\frac{V-W}{\lambda} + \mu \right) - 1 \right) - \frac{g_2 (V-W + \lambda \mu)^5}{1200} + \cdots \\
\hphantom{\Lambda(V,W,\lambda,\mu) =}{} - (V - W - \mu \lambda) \left(\log \left(\frac{V-W}{\lambda} - \mu \right) -1 \right) + \frac{g_2 (V-W - \lambda \mu)^5}{1200} + \cdots
\end{gather*}}
\caption{Q4 fact sheet. See Section \ref{sec-Q-3-leg} for the meaning of these functions.}\label{table-Q4-disc}
\end{table}

The continuum limit hierarchy is
\begin{gather*}
v_{3} = -6 v_{1}^{3} \wp(2 v) - \frac{3 v_{11}^{2}}{2 v_{1}} + v_{111} + \frac{3}{8 v_{1}}, \\
v_{5} = -90 v_{1}^{5} \wp(2 v)^{2} + 144 v_{1}^{5} \wp(v)^{2} - 24 v_{1}^{5} \wp''(v) + \frac{60 v_{1}^{3} v_{11} \wp(2 v) \wp''(v)}{\wp'(v)} - \frac{60 v_{1}^{3} v_{11} \wp(v) \wp''(v)}{\wp'(v)} \\
\hphantom{v_{5} =}{} + 60 v_{1}^{3} v_{11} \wp'(v) + 15 v_{1} v_{11}^{2} \wp(2 v) - 30 v_{1}^{2} v_{111} \wp(2 v) + \frac{5}{4} v_{1} \wp(2 v) - \frac{45 v_{11}^{4}}{8 v_{1}^{3}} + \frac{25 v_{11}^{2} v_{111}}{2 v_{1}^{2}} \\
\hphantom{v_{5} =}{} - \frac{5 v_{111}^{2}}{2 v_{1}} - \frac{5 v_{11} v_{1111}}{v_{1}} + v_{11111} + \frac{25 v_{11}^{2}}{16 v_{1}^{3}} - \frac{5 v_{111}}{8 v_{1}^{2}} - \frac{5}{128 v_{1}^{3}}, \\
\cdots
\end{gather*}
and $v_{2k} = 0$ for all $k \in \N$. To simplify these equations we used the doubling formula for the Weierstrass function,
\begin{gather*} \wp(2v) = -2 \wp(v) + \left( \frac{\wp''(v)}{2 \wp'(v)} \right)^2 .\end{gather*}

As the first nontrivial equation of the limit hierarchy we recognize the Krichever--Novikov equation~\cite{krichever1980holomorphic, nijhoff2002lax}. The continuum limit equations found in all previous examples can be considered as degenerations of the KN equation and its hierarchy, just like all other ABS equations of type Q can be obtained from~Q4~\cite{adler2004q4}.

The continuous pluri-Lagrangian 2-form has the coefficients
\begin{gather*}
\cL_{13} = 3 v_{1}^2 \wp(2 v_{}) + \frac{v_{11}^2}{4 v_{1}^2} - \frac{v_{3}}{4 v_{1}} + \frac{1}{16 v_{1}^2},\\
\cL_{15} = 15 v_{1}^4 \wp(2 v_{})^2 - 24 v_{1}^4 \wp(v_{})^2 + 4 v_{1}^4 \wp''(v_{}) - \frac{15}{2} v_{11}^2 \wp(2 v_{}) \\
\hphantom{\cL_{15} =}{} + \frac{3 v_{11}^4}{16 v_{1}^4} - \frac{v_{111}^2}{4 v_{1}^2} - \frac{v_{5}}{4 v_{1}} - \frac{5 v_{11}^2}{32 v_{1}^4} - \frac{1}{256 v_{1}^4} + \frac{5}{8} \wp(2 v_{}),\\
\cL_{35} = 45 v_{1}^6 \wp(2 v_{})^3 + 216 v_{1}^6 \wp(2 v_{}) \wp(v_{})^2 - 288 v_{1}^6 \wp(v_{})^3 - 36 v_{1}^6 \wp(2 v_{}) \wp''(v_{}) + 72 v_{1}^6 \wp(v_{}) \wp''(v_{})\\
\hphantom{\cL_{35} =}{}
 - 36 v_{1}^6 \wp'(v_{})^2 - \frac{18 v_{1}^4 v_{11} \wp(2 v_{})^2 \wp''(v_{})}{\wp'(v_{})} + \frac{18 v_{1}^4 v_{11} \wp(2 v_{}) \wp(v_{}) \wp''(v_{})}{\wp'(v_{})} - 18 v_{1}^4 v_{11} \wp(2 v_{}) \wp'(v_{})\\
\hphantom{\cL_{35} =}{}
 + \frac{63}{4} v_{1}^2 v_{11}^2 \wp(2 v_{})^2 - \frac{9}{2} v_{1}^3 v_{111} \wp(2 v_{})^2 + \frac{75}{2} v_{1}^3 v_{3} \wp(2 v_{})^2 - 54 v_{1}^2 v_{11}^2 \wp(v_{})^2 + 36 v_{1}^3 v_{111} \wp(v_{})^2 \\
\hphantom{\cL_{35} =}{}
 - 60 v_{1}^3 v_{3} \wp(v_{})^2 + 9 v_{1}^2 v_{11}^2 \wp''(v_{}) - 6 v_{1}^3 v_{111} \wp''(v_{}) + 10 v_{1}^3 v_{3} \wp''(v_{}) - \frac{135}{16} v_{1}^2 \wp(2 v_{})^2 \\
\hphantom{\cL_{35} =}{}
 + \frac{51}{2} v_{1}^2 \wp(v_{})^2 - \frac{81 v_{11}^3 \wp(2 v_{}) \wp''(v_{})}{2 \wp'(v_{})} + \frac{39 v_{1} v_{11} v_{111} \wp(2 v_{}) \wp''(v_{})}{\wp'(v_{})} - \frac{6 v_{1}^2 v_{1111} \wp(2 v_{}) \wp''(v_{})}{\wp'(v_{})} \\
\hphantom{\cL_{35} =}{}
 - \frac{15 v_{1} v_{11} v_{3} \wp(2 v_{}) \wp''(v_{})}{\wp'(v_{})} + \frac{81 v_{11}^3 \wp(v_{}) \wp''(v_{})}{2 \wp'(v_{})} - \frac{39 v_{1} v_{11} v_{111} \wp(v_{}) \wp''(v_{})}{\wp'(v_{})} \\
\hphantom{\cL_{35} =}{}
 + \frac{6 v_{1}^2 v_{1111} \wp(v_{}) \wp''(v_{})}{\wp'(v_{})} + \frac{15 v_{1} v_{11} v_{3} \wp(v_{}) \wp''(v_{})}{\wp'(v_{})} - \frac{81}{2} v_{11}^3 \wp'(v_{}) + 39 v_{1} v_{11} v_{111} \wp'(v_{}) \\
\hphantom{\cL_{35} =}{}
 - 6 v_{1}^2 v_{1111} \wp'(v_{}) - 15 v_{1} v_{11} v_{3} \wp'(v_{}) + \frac{495 v_{11}^4 \wp(2 v_{})}{16 v_{1}^2} - \frac{165 v_{11}^2 v_{111} \wp(2 v_{})}{4 v_{1}} + \frac{3}{4} v_{111}^2 \wp(2 v_{}) \\
\hphantom{\cL_{35} =}{}
 + \frac{27}{2} v_{11} v_{1111} \wp(2 v_{}) - \frac{3}{2} v_{1} v_{11111} \wp(2 v_{}) - 15 v_{11} v_{13} \wp(2 v_{}) + \frac{15 v_{11}^2 v_{3} \wp(2 v_{})}{4 v_{1}} \\
\hphantom{\cL_{35} =}{}
 + \frac{15}{2} v_{111} v_{3} \wp(2 v_{}) - \frac{9}{2} v_{1} v_{5} \wp(2 v_{}) - \frac{17}{4} v_{1}^2 \wp''(v_{}) + \frac{65 v_{11} \wp(2 v_{}) \wp''(v_{})}{8 \wp'(v_{})} - \frac{65 v_{11} \wp(v_{}) \wp''(v_{})}{8 \wp'(v_{})} \\
\hphantom{\cL_{35} =}{}
 + \frac{65}{8} v_{11} \wp'(v_{}) - \frac{255 v_{11}^2 \wp(2 v_{})}{32 v_{1}^2} - \frac{35 v_{111} \wp(2 v_{})}{16 v_{1}} + \frac{5 v_{3} \wp(2 v_{})}{16 v_{1}} + \frac{45 v_{11}^6}{64 v_{1}^6} - \frac{57 v_{11}^4 v_{111}}{32 v_{1}^5} \\
\hphantom{\cL_{35} =}{}
 + \frac{19 v_{11}^2 v_{111}^2}{16 v_{1}^4} - \frac{7 v_{111}^3}{8 v_{1}^3} + \frac{3 v_{11}^3 v_{1111}}{8 v_{1}^4} + \frac{3 v_{11} v_{111} v_{1111}}{4 v_{1}^3} - \frac{v_{1111}^2}{4 v_{1}^2} - \frac{3 v_{11}^2 v_{11111}}{8 v_{1}^3} + \frac{v_{111} v_{11111}}{4 v_{1}^2} \\
\hphantom{\cL_{35} =}{}
 - \frac{v_{111} v_{113}}{2 v_{1}^2} + \frac{3 v_{11}^3 v_{13}}{4 v_{1}^4} - \frac{v_{11} v_{111} v_{13}}{v_{1}^3} + \frac{v_{1111} v_{13}}{2 v_{1}^2} - \frac{v_{11} v_{15}}{2 v_{1}^2} + \frac{27 v_{11}^4 v_{3}}{32 v_{1}^5} - \frac{17 v_{11}^2 v_{111} v_{3}}{8 v_{1}^4}\\
\hphantom{\cL_{35} =}{}
 + \frac{7 v_{111}^2 v_{3}}{8 v_{1}^3} + \frac{3 v_{11} v_{1111} v_{3}}{4 v_{1}^3} - \frac{v_{11111} v_{3}}{4 v_{1}^2} - \frac{v_{11}^2 v_{5}}{8 v_{1}^3} + \frac{v_{111} v_{5}}{4 v_{1}^2} + \frac{15 \wp(2 v_{})}{256 v_{1}^2} - \frac{95 v_{11}^4}{256 v_{1}^6} \\
\hphantom{\cL_{35} =}{}
 + \frac{41 v_{11}^2 v_{111}}{64 v_{1}^5} - \frac{27 v_{111}^2}{64 v_{1}^4} - \frac{3 v_{11} v_{1111}}{32 v_{1}^4} + \frac{3 v_{11111}}{32 v_{1}^3} - \frac{5 v_{11} v_{13}}{16 v_{1}^4} - \frac{15 v_{11}^2 v_{3}}{64 v_{1}^5} + \frac{5 v_{111} v_{3}}{32 v_{1}^4} + \frac{v_{5}}{32 v_{1}^3}\\
\hphantom{\cL_{35} =}{}
 + \frac{55 v_{11}^2}{1024 v_{1}^6} - \frac{25 v_{111}}{512 v_{1}^5} + \frac{3 v_{3}}{512 v_{1}^5} - \frac{5}{4096 v_{1}^6}.
\end{gather*}

\section{ABS equations of type H}\label{sec-H}

For quad equations of type H we do not have a general strategy to find a form of the difference equation and its Lagrangian that is suitable for the continuum limit. This has to be investigated on a case-by-case basis. The continuum limit of equation H1 was discussed in \cite{vermeeren2019continuum}.

\subsection{H2}

No continuum limit is known for equation H2, which reads
\begin{gather}\label{H2-lambda}
(U - U_{12}) (U_1 - U_2) + \big(\lambda_2^2 - \lambda_1^2\big)( U + U_1 + U_2 + U_{12} ) + \lambda_2^4 - \lambda_1^4 = 0.
\end{gather}
We give a heuristic explanation of the obstruction one encounters when trying to pass to the continuum limit.

A tempting trick would be to change the sign of the field at every other vertex, $V(n,m) = (-1)^{n+m} U(n,m)$. The resulting equation
\begin{gather*} -(V - V_{12}) (V_1 - V_2) \pm \big(\lambda_2^2 - \lambda_1^2\big)( V - V_1 - V_2 + V_{12} ) + \lambda_2^4 - \lambda_1^4 = 0.
 \end{gather*}
has a suitable power series expansion, but the sign of the second term depends on the location in the lattice. In other words, the equation has become nonautonomous. Though there might not be any fundamental objection to this, the pluri-Lagrangian theory for nonautonomous systems has not yet been developed.

Of course the fact that one particular change of variables fails does not imply that there is no continuum limit. A better perspective on the issue is given by \emph{background solutions}. When the zero field $U \equiv 0$ is a solution to the quad equation, the continuum limit assumes the field to be small compared to the inverse of the parameters. If $U \equiv 0$ is not a solution, as is the case for equation~\eqref{H2-lambda}, one can look for a different solution to expand around. Such solutions are known as background solutions. There are two kinds of background solutions for H2 that treat both lattice directions equally~\cite{hietarinta2009soliton}:
\begin{gather*} U(n,m) = (\lambda_1 n + \lambda_2 m + c)^2 \end{gather*}
and
\begin{gather*} U(n,m) = \left(\frac{1}{2}(-1)^n \lambda_1 + \frac{1}{2}(-1)^m \lambda_2 + c \right)^2. \end{gather*}
Expanding around either of these solutions is equivalent to expanding around $0$ after a nonautonomous changes of variables. As in our first attempt above, these changes of variables make the equation nonautonomous.

Another approach would be to use the fact that H2 can be obtained as a degeneration of another ABS equation, for example Q2. However, this leads to the same problem. As one can see for example from equation (5.22) of \cite{nijhoff2009soliton}, solutions will pick up an alternating factor $(-1)^{m+n}$ in this degeneration.

\subsection[H3_{\delta=0}$]{$\boldsymbol{{\rm H3}_{\delta=0}}$}\label{sec-H3}

Equation H3 is usually given as
\begin{gather*}
\lambda_1 (U U_1 + U_2 U_{12}) - \lambda_2 (U U_2 + U_1 U_{12}) = 0,
\end{gather*}
but it is useful to perform a nonautonomous transformation $U(n,m) \rightarrow i^{n+m} U(n,m)$, which gives us
\begin{gather}\label{lmKdV}
\lambda_1 (U U_1 - U_2 U_{12}) - \lambda_2 (U U_2 - U_1 U_{12}) = 0 ,
\end{gather}
which is known as the \emph{lattice modified KdV equation} \cite{nijhoff1995discrete}. Note that even though we used a~nonautonomous transformation, the resulting equation is again autonomous, so we avoid the issues encountered with~H2. The continuum limit of equation~\eqref{lmKdV} is the hierarchy
\begin{gather*}
u_3 = u_{111} - 3 \frac{u_1 u_{11}}{u} , \\
u_{5} = - 10 \frac{u_{1}^{3} u_{11}}{u^{3}} + 10 \frac{2 u_{1} u_{11}^{2} + u_{1}^{2} u_{111}}{u^{2}} - 5 \frac{2
u_{11} u_{111} + u_{1} u_{1111}}{u} + u_{11111}, \\
\cdots
\end{gather*}
and $u_{2k} = 0$ for all $k \in \N$. However, we run in to trouble on the Lagrangian level. The Lagrangian given in~\cite{lobb2009lagrangian} for H3, adapted to the case $\delta = 0$, is
\begin{gather*}
L = \frac{1}{2} \log \left( \frac{U U_1}{\lambda_1} \right)^2 - \frac{1}{2} \log \left( \frac{U U_2}{\lambda_2} \right)^2 + \Li2 \left( \frac{\lambda_2 U_1}{\lambda_1 U_2} \right) - \Li2 \left( \frac{\lambda_1 U_1}{\lambda_2 U_2} \right)\\
\hphantom{L =}{} + 2 \big( \log(\lambda_1) - \log(\lambda_2) \big) \log(U) + 2 \log(\lambda_2) \big( \log(U_1) - \log(U_2) \big) ,
\end{gather*}
where $\Li2$ is the dilogarithm function. Unfortunately, the occurrence of expressions like $\log(\lambda_i)$ prohibits a power series expansion in the parameters~$\lambda_i$.

In order to find a better form of H3$_{\delta=0}$, with an expandable Lagrangian, we make the transformation $ U = \exp \big( \frac{i}{2} V \big)$. Then equation~\eqref{lmKdV} turns in to
\begin{gather*} \lambda_1 \big( {\rm e}^{ \frac{i}{2}(V + V_1) } - {\rm e}^{ \frac{i}{2}(V_2 + V_{12}) } \big) - \lambda_2 \big( {\rm e}^{ \frac{i}{2}(V + V_2) } - {\rm e}^{ \frac{i}{2}(V_1 + V_{12}) } \big) = 0. \end{gather*}
Multiplying with $\exp \big( {-}\frac{i}{4} (V+V_1+V_2+V_{12}) \big)$ turns this into
\begin{gather}\label{H3-SG}
\lambda_1 \sin \left( \frac{1}{4}(V+V_1-V_2-V_{12}) \right) - \lambda_2 \sin \left( \frac{1}{4}(V-V_1+V_2-V_{12}) \right) = 0,
\end{gather}
which is the form in which H3$_{\delta=0}$ arises from B\"acklund transformations for the sine-Gordon equation \cite[p.~60]{hietarinta2016discrete}. Additionally, it is closely related to the discrete sine-Gordon equation of Hirota \cite{hirota1977nonlinear3}. Using standard trigonometric identities we can rewrite the equation as
\begin{gather*} (\lambda_1 - \lambda_2) \tan \left( \frac{V_{12} - V}{4}\right) - (\lambda_1 + \lambda_2) \tan \left( \frac{V_1 - V_2}{4}\right) = 0, \end{gather*}
which we can put in a three-leg form:
\begin{gather*} \arctan \left( \frac{\lambda_1 + \lambda_2}{\lambda_1 - \lambda_2}\tan \left( \frac{V_1 - V_2}{4}\right) \right) = \frac{V_{12} - V_1}{4} - \frac{V - V_1}{4} . \end{gather*}
From the three-leg form we can derive a Lagrangian as in \cite{bobenko2010lagrangian}. It takes the implicit form
\begin{gather*} L = \frac{1}{8}(V_1 - V)^2 - \frac{1}{8}(V_2 - V)^2 - \mathcal{I}_{\lambda_1,\lambda_2}(V_1-V_2), \end{gather*}
where
\begin{gather*}\mathcal{I}_{\lambda_1,\lambda_2}(x) = \int_0^x \arctan \left( \frac{\lambda_1 + \lambda_2}{\lambda_1 - \lambda_2}\tan \left(\frac{y}{4}\right) \right) \d y .\end{gather*}
There is no need to evaluate this integral exactly. Instead one can expand the integrand as a~power series up to any desired order in $y$ and integrate this series. This is sufficient to write $L$ as a~power series up to the corresponding order in the parameters $\lambda_i$. A leading order calculation shows that this power series does not contain any terms of nonpositive order in either of the parameters, so it is suitable for the continuum limit.

In the continuum limit of equation~\eqref{H3-SG} we find the potential modified KdV hierarchy
\begin{gather*}
v_3 = v_{111} + \frac{1}{2} v_1^3 , \\
v_{5} = \frac{3}{8} v_{1}^{5} + \frac{5}{2} v_{1} v_{11}^{2} + \frac{5}{2} v_{1}^{2} v_{111} + v_{11111}, \\
\cdots
\end{gather*}
and $v_{2k} = 0$ for all $k \in \N$. It was shown in \cite{vermeeren2019continuum} that our continuum limit procedure is equivalent to the one presented for the lattice potential KdV and mKdV equations in \cite{wiersma1987lattice}, where it is proved that it results in the potential KdV and mKdV hierarchies.

Some coefficients of the pluri-Lagrangian $2$-form are
\begin{gather*}
\cL_{13} = \frac{1}{16} v_{1}^4 - \frac{1}{4} v_{11}^2 + \frac{1}{4} v_{} v_{13},\\
\cL_{15} = \frac{1}{32} v_{1}^6 - \frac{5}{8} v_{1}^2 v_{11}^2 + \frac{1}{4} v_{111}^2 + \frac{1}{4} v_{} v_{15},\\
\cL_{35} = -\frac{3}{256} v_{1}^8 + \frac{5}{32} v_{1}^4 v_{11}^2 - \frac{7}{32} v_{1}^5 v_{111} + \frac{3}{32} v_{1}^5 v_{3} + \frac{1}{16} v_{11}^4 - \frac{7}{8} v_{1} v_{11}^2 v_{111} - \frac{3}{8} v_{1}^2 v_{111}^2\\
\hphantom{\cL_{35} =}{}
 + \frac{3}{4} v_{1}^2 v_{11} v_{1111} - \frac{1}{8} v_{1}^3 v_{11111} - \frac{5}{4} v_{1}^2 v_{11} v_{13} + \frac{5}{8} v_{1} v_{11}^2 v_{3} + \frac{5}{8} v_{1}^2 v_{111} v_{3} - \frac{1}{8} v_{1}^3 v_{5} + \frac{1}{4} v_{1111}^2 \\
\hphantom{\cL_{35} =}{}
 - \frac{1}{4} v_{111} v_{11111} + \frac{1}{2} v_{111} v_{113} - \frac{1}{2} v_{1111} v_{13} + \frac{1}{2} v_{11} v_{15} + \frac{1}{4} v_{11111} v_{3} - \frac{1}{4} v_{111} v_{5}
\end{gather*}

\subsection[H3$_{\delta=1}$]{$\boldsymbol{{\rm H3}_{\delta=1}}$}

Due to difficulties analogous to those of equation H2, no continuum limit is known for equation H3$_{\delta=1}$,
\begin{gather*} \lambda_1(U U_1 + U_2 U_{12}) - \lambda_2(U U_2 + U_1 U_{12}) + \big(\lambda_1^2 - \lambda_2^2\big) = 0.\end{gather*}

\section{H3 and the sine-Gordon equation}\label{sec-SG}

Equation H3 is not just related to the modified KdV equation, but also to the sine-Gordon equation. In this section we investigate how to find both in the continuum limit. To do so we need to transform equation~\eqref{lmKdV} in a way that breaks the symmetry:
\begin{gather*} U(n,m) \mapsto U(n,m)^{(-1)^{m}} , \qquad \lambda_1 \mapsto \lambda_1 , \qquad \lambda_2 \mapsto \frac{1}{\lambda_2}.\end{gather*}
Setting again $ U = \exp \big( \frac{i}{2} V \big)$, we find the equation
\begin{gather}\label{SG}
\sin \left( \frac{1}{4}(V+V_1+V_2+V_{12}) \right) - \frac{1}{\lambda_1 \lambda_2} \sin \left( \frac{1}{4}(V-V_1-V_2+V_{12}) \right) = 0.
\end{gather}
In order to recover some kind of multidimensional consistency for this equation, we consider an even-dimensional lattice $\Z^{2N}$ with coordinates $(n_1,\dots,n_N,n_{\ol{1}},\dots, n_{\ol{N}})$, where the coordinates $(n_{\ol{1}},\dots, n_{\ol{N}})$ correspond to lattice directions where the above transformation takes place. Then on any quad spanned by directions $n_k$ and $n_{\ol{\ell}}$ we impose equation~\eqref{SG},
\begin{gather}\label{SG2}
\sin \left( \frac{1}{4}(V+V_k+V_{\ol{\ell}}+V_{k\ol{\ell}}) \right) - \frac{1}{\lambda_k \lambda_{\ol{\ell}}} \sin \left( \frac{1}{4}(V-V_k-V_{\ol{\ell}}+V_{k\ol{\ell}}) \right) = 0.
\end{gather}
The original equation H3, in the form of equation~\eqref{H3-SG}, is imposed on quads spanned by directions $n_k$ and $n_\ell$, and also on quads spanned by directions $n_{\ol{k}}$ and $n_{\ol{\ell}}$, because transforming both directions leaves the equation invariant.

We extend the Miwa correspondence, treating the barred variables in the same way as the variables without decoration, but contributing to a separate set of continuous variables $(t_{\ol{1}}, \dots, t_{\ol{N}})$,
\begin{gather*} t_i = (-1)^{i+1} \left(n_1 \frac{c \lambda_1^i}{i} + \dots + n_N \frac{c \lambda_N^i}{i} \right) + \tau_i, \\
 t_{\ol{i}} = (-1)^{i+1} \left(n_{\ol{1}} \frac{c \lambda_{\ol{1}}^i}{i} + \dots + n_{\ol{N}} \frac{c \lambda_{\ol{N}}^i}{i} \right) + \tau_{\ol{i}}. \end{gather*}
Then in the leading order of the continuum limit, equation~\eqref{SG2} produces the sine-Gordon equation
\begin{gather*} \sin(v) - v_{1 \ol{1}} = 0. \end{gather*}
In addition to this equation, the full series expansion in the parameters $\lambda_k$ and $\lambda_{\ol{\ell}}$ yields two copies of the potential modified KdV hierarchy,
\begin{alignat*}{3}
&v_2= 0 , \qquad && v_{\ol{2}} = 0 ,& \\
&v_3= v_{111} + \frac{1}{2} v_1^3 , \qquad && v_{\ol{3}}= v_{\ol{111}} + \frac{1}{2} v_{\ol{1}}^3 ,& \\
&v_4= 0 , \qquad && v_{\ol{4}}= 0 , &\\
& v_{5}= \frac{3}{8} v_{1}^{5} + \frac{5}{2} v_{1} v_{11}^{2} + \frac{5}{2} v_{1}^{2} v_{111} + v_{11111} , \qquad&&
v_{\ol{5}}= \frac{3}{8} v_{\ol{1}}^{5} + \frac{5}{2} v_{\ol{1}} v_{\ol{11}}^{2} +
\frac{5}{2} v_{\ol{1}}^{2} v_{\ol{111}} + v_{\ol{11111}} ,& \\
&\cdots &&\cdots s&
\end{alignat*}
These are the same hierarchies we obtain when taking the continuum limits of the equations imposed on quads spanned by directions $n_k$ and $n_\ell$, or $n_{\ol{k}}$ and $n_{\ol{\ell}}$.

The discrete Lagrangian takes slightly different forms on the three types of quads:
\begin{gather*}
L(\square_{k\ell}) = \frac{1}{8}(V_k - V)^2 - \frac{1}{8}(V_\ell - V)^2 - \mathcal{I}_{\lambda_k,\lambda_\ell}(V_k-V_\ell) , \\
L(\square_{k\ol{\ell}}) = \frac{1}{8}(V_k - V)^2 - \frac{1}{8}(V_{\ol{\ell}} + V)^2 - \mathcal{J}_{\lambda_k,\lambda_{\ol{\ell}}}(V_k+V_{\ol{\ell}}) , \\
L(\square_{\ol{k\ell}}) = \frac{1}{8}(V_{\ol{k}} + V)^2 - \frac{1}{8}(V_{\ol{\ell}} + V)^2 + \mathcal{I}_{\lambda_{\ol{k}},\lambda_{\ol{\ell}}}(V_{\ol{k}}-V_{\ol{\ell}}) ,
\end{gather*}
where
\begin{gather*} \mathcal{I}_{\lambda_k,\lambda_\ell}(x) = \int_0^x \arctan \left( \frac{\lambda_k + \lambda_\ell}{\lambda_k - \lambda_\ell}\tan \left(\frac{y}{4}\right) \right) \d y \end{gather*}
and
\begin{gather*} \mathcal{J}_{\lambda_k,\lambda_{\ol{\ell}}}(x) =
\mathcal{I}_{\lambda_k,\frac{1}{\lambda_{\ol{\ell}}}}(x) = \int_0^x \arctan \left( \frac{\lambda_k \lambda_{\ol{\ell}} + 1}{\lambda_k \lambda_{\ol{\ell}} - 1}\tan \left(\frac{y}{4}\right) \right) \d y .\end{gather*}
To get the required leading-order cancellation, we add to $L$ the exterior derivative of the discrete 1-form $\eta$ given by $\eta(V,V_{\ol{\ell}}) = -\frac{1}{8} \big( V^2 + V_{\ol{\ell}}^2 \big)$ and $\eta(V,V_k) = 0$. This yields
\begin{gather*}
L(\square_{k\ell}) = \frac{1}{8}(V_k - V)^2 - \frac{1}{8}(V_\ell - V)^2 - \mathcal{I}_{\lambda_k,\lambda_\ell}(V_k-V_\ell) , \\
L(\square_{k\ol{\ell}}) = \frac{1}{8}(V_k - V)^2 - \frac{1}{8}(V_{\ol{\ell}} + V)^2 - \mathcal{J}_{\lambda_k,\lambda_{\ol{\ell}}}(V_k+V_{\ol{\ell}}) - \frac{1}{8} \big( V_k^2 - V^2 + V_{k\ol{\ell}}^2 - V_{\ol{\ell}}^2 \big), \\
L(\square_{\ol{k\ell}}) = \frac{1}{8}(V_{\ol{k}} + V)^2 - \frac{1}{8}(V_{\ol{\ell}} + V)^2 + \mathcal{I}_{\lambda_{\ol{k}},\lambda_{\ol{\ell}}}(V_{\ol{k}}-V_{\ol{\ell}}) + \frac{1}{4} \big( V_{\ol{\ell}}^2 - V_{\ol{k}}^2 \big) \\
\hphantom{L(\square_{\ol{k\ell}})}{}= -\frac{1}{8} (V_{\ol{k}} - V)^2 + \frac{1}{8}(V_{\ol{\ell}} - V)^2 + \mathcal{I}_{\lambda_{\ol{k}},\lambda_{\ol{\ell}}}(V_{\ol{k}} - V_{\ol{\ell}}) .
\end{gather*}
The resulting pluri-Lagrangian 2-form is of the form
\begin{gather*} \cL = \sum_{i,j \in \{1,\ol{1},2,\ol{2},\dots\}} \cL_{ij} \, \d t_i \wedge \d t_j , \end{gather*}
where $\cL_{ij} = - \cL_{ji}$. If $i,j \in \{1,2,\dots\}$, then the coefficients are exactly those found in Section~\ref{sec-H3}. If $i,j \in \{\ol{1},\ol{2},\dots\}$ the coefficients are obtained from those by placing a bar over every derivative and putting a minus sign in front. As before, we can restrict our attention to the odd-numbered times to get a pluri-Lagrangian structure for the nontrivial equations
\begin{gather*} \cL = \sum_{i,j \in \{1,\ol{1},3,\ol{3},\dots\}} \cL_{ij} \,\d t_i \wedge \d t_j . \end{gather*}
Some coefficients of mixed type are
\begin{gather*}
\cL_{\ol{1}1} = -\frac{1}{4} v_{1\ol{1}} v_{} - \frac{1}{2} \cos(v),\\
\cL_{\ol{1}3} = -\frac{1}{4} v_{1}^2 \cos(v) - \frac{1}{2} v_{11} v_{1\ol{1}} + \frac{1}{4} v_{3} v_{\ol{1}} + \frac{1}{2} v_{11} \sin(v),\\
\cL_{\ol{1}5} = -\frac{3}{16} v_{1}^4 \cos(v) - \frac{5}{4} v_{11} v_{1}^2 v_{1\ol{1}} + \frac{3}{4} v_{11} v_{1}^2 \sin(v) + \frac{1}{4} v_{11}^2 \cos(v) - \frac{1}{2} v_{111} v_{1} \cos(v)\\
\hphantom{\cL_{\ol{1}5} =}{}
+ \frac{1}{2} v_{111} v_{11\ol{1}} - \frac{1}{2} v_{1111} v_{1\ol{1}} + \frac{1}{4} v_{5} v_{\ol{1}} + \frac{1}{2} v_{1111} \sin(v),\\
\cL_{\ol{3}1} = -\frac{1}{4} v_{\ol{1}}^2 \cos(v) - \frac{1}{2} v_{1\ol{3}} v_{} + \frac{1}{2} v_{1\ol{1}\ol{1}} v_{\ol{1}} - \frac{1}{4} v_{1} v_{\ol{3}} + \frac{1}{2} v_{\ol{1}\ol{1}} \sin(v),\\
\cL_{\ol{3}3} = -\frac{1}{8} v_{1}^2 v_{\ol{1}}^2 \cos(v) + \frac{1}{4} v_{11} v_{\ol{1}}^2 \sin(v) + \frac{1}{4} v_{1}^2 v_{\ol{1}\ol{1}} \sin(v) + \frac{1}{2} v_{11} v_{\ol{1}\ol{1}} \cos(v) - \frac{1}{2} v_{11} v_{1\ol{3}} \\
\hphantom{\cL_{\ol{3}3} =}{}
- \frac{1}{4} v_{3\ol{3}} v_{}- \frac{1}{2} v_{1} v_{\ol{1}} + \frac{1}{2} v_{3\ol{1}\ol{1}} v_{\ol{1}} - \frac{1}{2} \cos(v),\\
\cL_{\ol{3}5} = -\frac{3}{32} v_{1}^4 v_{\ol{1}}^2 \cos(v) + \frac{3}{8} v_{11} v_{1}^2 v_{\ol{1}}^2 \sin(v) + \frac{3}{16} v_{1}^4 v_{\ol{1}\ol{1}} \sin(v) + \frac{1}{8} v_{11}^2 v_{\ol{1}}^2 \cos(v) \\
\hphantom{\cL_{\ol{3}5} =}{}- \frac{1}{4} v_{111} v_{1} v_{\ol{1}}^2 \cos(v) + \frac{3}{4} v_{11} v_{1}^2 v_{\ol{1}\ol{1}} \cos(v) - \frac{5}{4} v_{11} v_{1}^2 v_{1\ol{3}} - \frac{1}{4} v_{1}^3 v_{\ol{1}} + \frac{1}{4} v_{1111} v_{\ol{1}}^2 \sin(v)\\
\hphantom{\cL_{\ol{3}5} =}{} - \frac{1}{4} v_{11}^2 v_{\ol{1}\ol{1}} \sin(v) + \frac{1}{2} v_{111} v_{1} v_{\ol{1}\ol{1}} \sin(v) - \frac{1}{4} v_{1}^2 \cos(v) + \frac{1}{2} v_{1111} v_{\ol{1}\ol{1}} \cos(v) \\
\hphantom{\cL_{\ol{3}5} =}{}
 + \frac{1}{2} v_{111} v_{11\ol{3}} - \frac{1}{2} v_{1111} v_{1\ol{3}} - \frac{1}{4} v_{5\ol{3}} v_{} - \frac{1}{2} v_{111} v_{\ol{1}} + \frac{1}{2} v_{5\ol{1}\ol{1}} v_{\ol{1}} + \frac{1}{2} v_{11} \sin(v),\\
\cL_{\ol{5}1} = -\frac{7}{144} v_{\ol{1}}^4 \cos(v) + \frac{5}{36} v_{1\ol{1}\ol{1}} v_{\ol{1}}^3 + \frac{5}{12} v_{1\ol{1}} v_{\ol{1}}^2 v_{\ol{1}\ol{1}} - \frac{1}{12} v_{\ol{1}}^2 v_{\ol{1}\ol{1}} \sin(v) + \frac{1}{4} v_{\ol{1}\ol{1}}^2 \cos(v) \\
\hphantom{\cL_{\ol{5}1} =}{}
- \frac{2}{9} v_{\ol{1}} v_{\ol{1}\ol{1}\ol{1}} \cos(v) - \frac{5}{18} v_{\ol{1}} v_{\ol{3}} \cos(v) - \frac{1}{2} v_{1\ol{5}} v_{} - \frac{1}{18} v_{1\ol{1}\ol{1}\ol{1}\ol{1}} v_{\ol{1}} + \frac{5}{9} v_{1\ol{1}\ol{3}} v_{\ol{1}} - \frac{2}{9} v_{1\ol{1}\ol{1}\ol{1}} v_{\ol{1}\ol{1}} \\
\hphantom{\cL_{\ol{5}1} =}{}
- \frac{5}{18} v_{1\ol{3}} v_{\ol{1}\ol{1}} - \frac{1}{4} v_{1} v_{\ol{5}} - \frac{1}{18} v_{\ol{1}\ol{1}\ol{1}\ol{1}} \sin(v) + \frac{5}{9} v_{\ol{1}\ol{3}} \sin(v),\\
\cL_{\ol{5}3} = -\frac{7}{288} v_{1}^2 v_{\ol{1}}^4 \cos(v) + \frac{7}{144} v_{11} v_{\ol{1}}^4 \sin(v) - \frac{1}{24} v_{1}^2 v_{\ol{1}}^2 v_{\ol{1}\ol{1}} \sin(v) - \frac{1}{12} v_{11} v_{\ol{1}}^2 v_{\ol{1}\ol{1}} \cos(v)\\
\hphantom{\cL_{\ol{5}3} =}{}
 + \frac{1}{8} v_{1}^2 v_{\ol{1}\ol{1}}^2 \cos(v) - \frac{1}{9} v_{1}^2 v_{\ol{1}} v_{\ol{1}\ol{1}\ol{1}} \cos(v) - \frac{5}{36} v_{1}^2 v_{\ol{1}} v_{\ol{3}} \cos(v) - \frac{1}{9} v_{1} v_{\ol{1}}^3 + \frac{5}{36} v_{3\ol{1}\ol{1}} v_{\ol{1}}^3 \\
\hphantom{\cL_{\ol{5}3} =}{}
 + \frac{5}{12} v_{3\ol{1}} v_{\ol{1}}^2 v_{\ol{1}\ol{1}} - \frac{1}{4} v_{11} v_{\ol{1}\ol{1}}^2 \sin(v) + \frac{2}{9} v_{11} v_{\ol{1}} v_{\ol{1}\ol{1}\ol{1}} \sin(v) - \frac{1}{36} v_{1}^2 v_{\ol{1}\ol{1}\ol{1}\ol{1}} \sin(v) \\
\hphantom{\cL_{\ol{5}3} =}{}
+ \frac{5}{18} v_{1}^2 v_{\ol{1}\ol{3}} \sin(v) + \frac{5}{18} v_{11} v_{\ol{1}} v_{\ol{3}} \sin(v) - \frac{1}{4} v_{\ol{1}}^2 \cos(v) - \frac{1}{18} v_{11} v_{\ol{1}\ol{1}\ol{1}\ol{1}} \cos(v)\\
\hphantom{\cL_{\ol{5}3} =}{}
 + \frac{5}{9} v_{11} v_{\ol{1}\ol{3}} \cos(v) - \frac{1}{2} v_{11} v_{1\ol{5}}
 - \frac{1}{4} v_{3\ol{5}} v_{} - \frac{1}{18} v_{3\ol{1}\ol{1}\ol{1}\ol{1}} v_{\ol{1}} + \frac{5}{9} v_{3\ol{1}\ol{3}} v_{\ol{1}} \\
\hphantom{\cL_{\ol{5}3} =}{}
 - \frac{2}{9} v_{3\ol{1}\ol{1}\ol{1}} v_{\ol{1}\ol{1}} - \frac{5}{18} v_{3\ol{3}} v_{\ol{1}\ol{1}} - \frac{2}{9} v_{1} v_{\ol{1}\ol{1}\ol{1}} - \frac{5}{18} v_{1} v_{\ol{3}}
 + \frac{1}{2} v_{\ol{1}\ol{1}} \sin(v).
\end{gather*}

A few examples of multi-time Euler--Lagrange equations are
\begin{gather*}
0 = \var{\ol{1}1}{\cL_{\ol{1}1}}{v} = -\tfrac{1}{2} v_{1\ol{1}} + \tfrac{1}{2} \sin(v) , \\
0 = \var{\ol{1}3}{\cL_{\ol{1}3}}{v_{\ol{1}}} - \var{\ol{3}3}{\cL_{\ol{3}3}}{v_{\ol{3}}}
= \frac{1}{4} v_3 - \frac{1}{4} \big( v_1^3 + 2 v_{111} - v_3 \big) , \\
0 = \var{1\ol{3}}{\cL_{1\ol{3}}}{v_1} - \var{3\ol{3}}{\cL_{3\ol{3}}}{v_3}
= -\frac{1}{4} v_{\ol{3}} + \frac{1}{4} \big( v_{\ol{1}}^3 + 2 v_{\ol{111}} - v_{\ol{3}} \big) , \\
0 = \var{\ol{1}5}{\cL_{\ol{1}5}}{v_{\ol{1}}} - \var{15}{\cL_{15}}{v_1} = \frac{1}{4} v_5 - \left( \frac{3}{16} v_1^5 + \frac{5}{4} v_1 v_{11}^2 + \frac{5}{4} v_1^2 v_{111} + \frac{1}{2} v_{11111} - \frac{1}{4} v_5 \right) .
\end{gather*}
One recognizes the equations of the double hierarchy. All multi-time Euler--Lagrange equations are either trivial, or differential consequences of the equations of the hierarchy.

\section{About the even-numbered times}\label{sec-even}

In all the examples so far, only the odd-numbered times have nontrivial equations associated to them. This is consistent with several purely continuous descriptions of integrable hierarchies. From the perspective of continuum limits, it follows from the fact that for all equations we dealt with, $Q$ is an even or odd function of the parameters.

\begin{Proposition}\label{prop-time-reversal}If the difference equation $Q = 0$ satisfies
\begin{gather*} Q(U,U_1,U_2,U_{12},\lambda_1,\lambda_2) = \pm Q(U,U_1,U_2,U_{12},-\lambda_1,-\lambda_2) \end{gather*}
then the continuum limit hierarchy is invariant under simultaneous reversal of all even-numbered times, $t_{2k} \mapsto -t_{2k}$.
\end{Proposition}

\begin{proof}Due to the symmetry assumptions on $Q$, we have
\begin{gather*} Q(U,U_1,U_2,U_{12},\lambda_1,\lambda_2) = \pm Q(U_{12},U_2,U_1,U,-\lambda_1,-\lambda_2) . \end{gather*}
Let $u(t_1,t_2,\dots)$ be a solution to the continuum limit equations. Then
\begin{gather}
Q\bigg(u(t_1,t_2,\dots), u \left(t_1 + c \lambda_1, t_2 - c \frac{\lambda_1^2}{2},\dots \right),
u \left(t_1 + c \lambda_2, t_2 - c \frac{\lambda_2^2}{2},\dots \right), \nonumber\\
 \qquad u \left(t_1 + c \lambda_1 + c \lambda_2, t_2 - c \frac{\lambda_1^2}{2} - c \frac{\lambda_2^2}{2},\dots \right),
\lambda_1,\lambda_2 \bigg) = 0.\label{before-reversal}
\end{gather}
Now consider the same equation at times shifted once in both lattice directions, $\tau_i = t_i + (-1)^{i+1} c \big( \frac{\lambda_1^i}{i} + \frac{\lambda_2^i}{i} \big)$. Due to the symmetry of $Q$ we have that
\begin{gather*}
Q\bigg( u(\tau_1,\tau_2,\dots), u \left(\tau_1 - c \lambda_1, \tau_2 + c \frac{\lambda_1^2}{2},\dots \right),u \left(\tau_1 - c \lambda_2, \tau_2 + c \frac{\lambda_2^2}{2},\dots \right), \\
\qquad u \left(\tau_1 - c \lambda_1 - c \lambda_2, \tau_2 + c \frac{\lambda_1^2}{2} + c \frac{\lambda_2^2}{2},\dots \right),-\lambda_1,-\lambda_2 \bigg) = 0.
\end{gather*}
Introducing the parameters $\mu_1 = - \lambda_1$ and $\mu_2 = - \lambda_2$ we find
\begin{gather}
Q\bigg(u(\tau_1,\tau_2,\dots), u \left(\tau_1 + c \mu_1, \tau_2 + c \frac{\mu_1^2}{2},\dots \right),u \left(\tau_1 + c \mu_2, \tau_2 + c \frac{\mu_2^2}{2},\dots \right), \nonumber\\
\qquad u \left(\tau_1 + c \mu_1 + c \mu_2, \tau_2 + c \frac{\mu_1^2}{2} + c \frac{\mu_2^2}{2},\dots \right),\mu_1,\mu_2 \bigg) = 0 .\label{after-reversal}
\end{gather}
Comparing equations \eqref{before-reversal} and~\eqref{after-reversal}, we immediately see that their series expansions in~$\lambda_1$,~$\lambda_2$ respectively $\mu_1$, $\mu_2$ only differ by a minus sign for each derivative with respect to an even-numbered time $t_{2k}$. In other words, if we have a~solution~$u$ to the continuum limit hierarchy, then reversing all even times gives a new solution.
\end{proof}

\begin{Corollary}\label{cor-odd}If the continuum limit hierarchy of the difference equation $Q = 0$, as in Proposition~{\rm \ref{prop-time-reversal}}, consists of evolutionary equations $u_k = f_k(u,u_1,u_{11},\dots)$, then for all even $k$ we have $f_k = 0$.
\end{Corollary}

This property can also be understood by considering the plane wave factors
\begin{gather*} \rho_k(\bn) = \prod_{i=1}^N \left( \frac{\lambda_i - k}{\lambda_i + k} \right), \end{gather*}
for which a discrete shift corresponds to a multiplication by
\begin{gather*} \frac{\lambda_i - k}{\lambda_i + k} = \exp \left( \log \left( 1 - \frac{k}{\lambda_i} \right) - \log \left(1 + \frac{k}{\lambda_i} \right) \right) = \exp \left( -2 \frac{k}{\lambda_i} - \frac{2}{3} \frac{k^3}{\lambda_i^3} - \cdots \right) .\end{gather*}
The series expansion in the exponential only contains odd powers of the parameters, hence they will only involve the odd-numbered times in the continuum limit. Plane wave factors play a~role in several works on continuum limits \cite{quispel1984linear, wiersma1987lattice, yoo2011discrete} and in the construction of soliton solutions \cite{hietarinta2009soliton,nijhoff2009soliton} and rational solutions~\cite{zhang2017rational, zhao2019rational}. It is shown in \cite[Appendix~C]{zhang2017rational} that all rational solutions depend only on odd powers of the parameters.

Unfortunately, as we will see in Section~\ref{sec-gd}, there are situations where the difference equation involves more than four points and the continuum limit contains PDEs that are not evolutionary, so Proposition~\ref{prop-time-reversal} and Corollary~\ref{cor-odd} do not apply. Here we will see a different pattern of trivial equations. This pattern can be understood in the framework of reductions of the lattice KP system \cite{fu2017reductions,fu2018linear}, or from the construction by pseudodifferential operators of the Gelfand--Dickey hierarchy~\cite{dickey2003soliton}.

\section{Gelfand--Dickey hierarchies}\label{sec-gd}

A discrete counterpart of the Gelfand--Dickey hierarchy was introduced in \cite{nijhoff1992lattice}, its pluri-Lag\-ran\-gian structure in \cite{lobb2010lagrangian}. The $N$-th member of the discrete hierarchy is a system of quad equations with $2N-3$ components, which we denote by $V^{[1]}, \dots, V^{[N-2]}$, $W^{[1]}, \dots, W^{[N-2]}$, and $U = V^{[0]} = W^{[0]}$. The equations are
\begin{gather}
V_2^{[j+1]} - V_1^{[j+1]} = \left( \frac{1}{\lambda_1} - \frac{1}{\lambda_2} + U_2 - U_1 \right) V_{12}^{[j]} - \frac{1}{\lambda_1} V_2^{[j]} + \frac{1}{\lambda_2} V_1^{[j]}, \label{lGD1}\\
W_2^{[j+1]} - W_1^{[j+1]} = -\left( \frac{1}{\lambda_1} - \frac{1}{\lambda_2} + U_2 - U_1 \right) W^{[j]} - \frac{1}{\lambda_2} W_2^{[j]} + \frac{1}{\lambda_1} W_1^{[j]} , \label{lGD2}
\end{gather}
for $j = 0, \dots, N-3$, and
\begin{gather}
V_{12}^{[N-2]} - W^{[N-2]} = \frac{\frac{1}{\lambda_1^N} - \frac{1}{\lambda_2^N}}{ \frac{1}{\lambda_1} - \frac{1}{\lambda_2} + U_2 - U_1} - \gamma_{N-1}
 + \sum_{i = 0}^{N-3} \sum_{j=0}^{N-3-i} \gamma_{N-3-i-j} V_{12}^{[j]} W^{[i]} \nonumber\\
\hphantom{V_{12}^{[N-2]} - W^{[N-2]} =}{} - \sum_{j=0}^{N-3} \gamma_{N-2-j} \big( V_{12}^{[j]} - W^{[j]} \big),\label{lGD3}
\end{gather}
where
\begin{gather*} \gamma_j = (-1)^j \left( \frac{1}{\lambda_1^j} + \frac{1}{\lambda_1^{j-1} \lambda_2} + \dots + \frac{1}{\lambda_1 \lambda_2^{j-1}} + \frac{1}{\lambda_2^j} \right) .\end{gather*}

It is important to note that the variables $V^{[N-2]}$ and $W^{[N-2]}$ can be eliminated. The reduced system consists of equations \eqref{lGD1}--\eqref{lGD2} for $j = 0, \dots, N-4$ and the 9-point equation
\begin{gather}
\left( \frac{1}{\lambda_1} - \frac{1}{\lambda_2} + U_{122} - U_{112} \right) V_{1122}^{[N-3]} - \frac{1}{\lambda_1} V_{122}^{[N-3]} + \frac{1}{\lambda_2} V_{112}^{[N-3]} \nonumber\\
\qquad\quad{}+ \left( \frac{1}{\lambda_1} - \frac{1}{\lambda_2} + U_2 - U_1 \right) W^{[N-3]} + \frac{1}{\lambda_2} W_2^{[N-3]} - \frac{1}{\lambda_1} W_1^{[N-3]} \nonumber\\
\qquad{} = \frac{\frac{1}{\lambda_1^N} - \frac{1}{\lambda_2^N}}{ \frac{1}{\lambda_1} - \frac{1}{\lambda_2} + U_{22} - U_{12}}
+ \sum_{i = 0}^{N-3} \sum_{j=0}^{N-3-i} \gamma_{N-3-i-j} \big( V_{122}^{[j]} W_2^{[i]} - V_{112}^{[j]} W_1^{[i]} \big) \nonumber\\
\qquad\quad {}- \frac{\frac{1}{\lambda_1^N} - \frac{1}{\lambda_2^N}}{ \frac{1}{\lambda_1} - \frac{1}{\lambda_2} + U_{12} - U_{11}} - \sum_{j=0}^{N-3} \gamma_{N-2-j} \big( V_{122}^{[j]} - W_2^{[j]} - V_{112}^{[j]} + W_1^{[j]} \big),\label{lGD3-9}
\end{gather}
obtained by evaluating $V_{122}^{[N-2]} - W_2^{[N-2]} - V_{112}^{[N-2]} + W_1^{[N-2]}$ once with equations \eqref{lGD1}--\eqref{lGD2} and once with equation~\eqref{lGD3}.

The Lagrangian given in \cite{lobb2010lagrangian} for the $N$-th lattice Gelfand--Dickey equation is
\begin{gather*}
L = (-1)^{N+1} \left( \frac{1}{\lambda_1^N} - \frac{1}{\lambda_2^N} \right) \log \left(\frac{1}{\lambda_1} - \frac{1}{\lambda_2} - U_1 + U_2 \right)- \gamma_{N-1} (U_2 - U_1) \\
\hphantom{L =}{} - \sum_{j=0}^{N-2} \gamma_{N-2-j} (U_2-U_1) V_{12}^{[j]} - \sum_{j=1}^{N-2} \sum_{i=0}^{N-2-j}
\gamma_{N-2-i-j} W^{[i]} \big( V_2^{[j]} - V_1^{[j]} \big) \\
\hphantom{L =}{} + \sum_{j=0}^{N-3} \sum_{i=0}^{N-3-j} \gamma_{N-3-i-j} W^{[i]} \left( \left(\frac{1}{\lambda_1} - \frac{1}{\lambda_2} - U_1 + U_2 \right) V_{12}^{[j]} - \frac{1}{\lambda_1} V_2^{[j]} + \frac{1}{\lambda_2} V_1^{[j]} \right).
\end{gather*}
Note that this Lagrangian depends on the field $V^{[N-2]}$ but not on $W^{[N-2]}$. A more symmetric equivalent Lag\-rangian is obtained by adding the exact discrete differential form
\begin{gather*} -U V_1^{[N-2]} - U_1 V_{12}^{[N-2]} + U_2 V_{12}^{[N-2]} + U V_{2}^{[N-2]}, \end{gather*}
which does not change the Euler--Lagrange equations (see Proposition~\ref{prop-exact}). The resulting Lagrangian does not depend on $V^{[N-2]}$ either,
\begin{gather}
L = (-1)^{N+1} \left( \frac{1}{\lambda_1^N} - \frac{1}{\lambda_2^N} \right) \log \left(\frac{1}{\lambda_1} - \frac{1}{\lambda_2} - U_1 + U_2 \right)- \gamma_{N-1} (U_2 - U_1) \nonumber\\
\hphantom{L =}{} - \sum_{j=0}^{N-3} \gamma_{N-2-j} (U_2-U_1) V_{12}^{[j]} - \sum_{j=1}^{N-3} \sum_{i=0}^{N-2-j}
\gamma_{N-2-i-j} W^{[i]} \big( V_2^{[j]} - V_1^{[j]} \big) \nonumber\\
\hphantom{L =}{} + \sum_{j=0}^{N-3} \sum_{i=0}^{N-3-j} \gamma_{N-3-i-j} W^{[i]} \left( \left(\frac{1}{\lambda_1} - \frac{1}{\lambda_2} - U_1 + U_2 \right) V_{12}^{[j]} - \frac{1}{\lambda_1} V_2^{[j]} + \frac{1}{\lambda_2} V_1^{[j]} \right).\label{GD-lag}
\end{gather}
Now \looseness=-1 that both variables $V^{[N-2]}$ and $W^{[N-2]}$ are absent from the Lagrangian, it becomes clear that the variational formulation does not produce the quad version of the equations \eqref{lGD1}--\eqref{lGD3}, but rather the 9-point version~\eqref{lGD3-9}. In particular this means that, from the Lagrangian perspective, the scalar form of the Boussinesq equation ($N=3$) is the most natural. The first truly multi-component Lagrangian equation of the hierarchy is found for $N=4$. The continuum limit of this equation will allow us to formulate a continuous multi-component pluri-Lagrangian system.

Note that the Lagrangian \eqref{GD-lag} depends on fields on a single quad and on the corresponding lattice parameters. Hence it fits the pluri-Lagrangian theory. The multi-dimensional consistency of the equations can either be checked as consistency around an elementary cube for the quad version of the equation, or in a 27-point cube for the 9-point formulation. For the Boussinesq equation the former was explicitly done in \cite{tongas2005boussinesq} and the latter was numerically verified in \cite[Section~5.7]{walker2001similarity}. The multidimensional consistency also follows from the construction of the GD hierarchy using the direct linearization method \cite{nijhoff1992lattice}. In this construction, lattice shifts are identified with B\"acklund transformations, hence the permutability property implies multidimensional consistency.

\subsection{Continuum limit of the lattice Boussinesq equation (GD3)}

We consider the lattice Boussinesq equation in its 9-point scalar form, i.e., equation \eqref{lGD3-9} for $N=3$,
\begin{gather*}
\frac{ \frac{1}{\lambda_1^3} - \frac{1}{\lambda_2^3} }{ \frac{1}{\lambda_1} - \frac{1}{\lambda_2} + U_{12} - U_{11} } - \frac{ \frac{1}{\lambda_1^3} - \frac{1}{\lambda_2^3} }{ \frac{1}{\lambda_1} - \frac{1}{\lambda_2} + U_{22} - U_{12} } - U_2 U_{122} + U_1 U_{112} \\
\qquad{} + \left( \frac{1}{\lambda_1} - \frac{1}{\lambda_2} + U_{122} - U_{112} \right) U_{1122} + \left( \frac{1}{\lambda_1} - \frac{1}{\lambda_2} + U_{2} - U_{1} \right) U \\
\qquad{} - \left( \frac{2}{\lambda_1} + \frac{1}{\lambda_2} \right) (U_1 + U_{122}) + \left( \frac{1}{\lambda_1} + \frac{2}{\lambda_2} \right) (U_2 + U_{112})
= 0 .
\end{gather*}
For this equation we use the Miwa correspondence \eqref{miwa} with $c = -3$. As always, we perform a~double series expansion of the lattice equation
\begin{gather*} \sum_{i,j = 0}^\infty \cF_{ij} \lambda_1^i \lambda_2^j = 0.\end{gather*}
The first column of coefficients of this expansion is
\begin{gather*}
\cF_{00} = 0 ,\\
\cF_{10} = 18 v_{1} v_{11} - \frac{9}{2} v_{1111} - \frac{3}{2} v_{22} ,\\
\cF_{20} = 81 v_{11}^{2} + 81 v_{1} v_{111} - \frac{81}{4} v_{11111} - \frac{27}{4} v_{122} - 3 v_{23} ,\\
\cF_{30} = -108 v_{1}^{2} v_{11} + 648 v_{11} v_{111} + 243 v_{1} v_{1111} + 54 v_{1} v_{112} + 54 v_{11} v_{12} - 18 v_{12} v_{2} - 9 v_{1} v_{22} \\
\hphantom{\cF_{30} =}{} - 54 v_{111111} - \frac{27}{2} v_{11112} - 9 v_{1122} - 12 v_{123} - \frac{9}{2} v_{222} - 3 v_{24} - \frac{4}{3} v_{33} ,\\
\cF_{40} = -810 v_{1} v_{11}^{2} - 405 v_{1}^{2} v_{111} - 135 v_{1}^{2} v_{12} - 135 v_{1} v_{11} v_{2} + 1215 v_{111}^{2} + \frac{6885}{4} v_{11} v_{1111} \\
\hphantom{\cF_{40} =}{}
+ \frac{2025}{4} v_{1} v_{11111} + \frac{675}{2} v_{1} v_{1112} + \frac{1215}{2} v_{11} v_{112} + 45 v_{1} v_{113} + \frac{675}{2} v_{111} v_{12} - \frac{135}{2} v_{12}^{2} \\
\hphantom{\cF_{40} =}{}
- \frac{135}{4} v_{1} v_{122} + 45 v_{11} v_{13} + \frac{135}{4} v_{1111} v_{2} - \frac{135}{2} v_{112} v_{2} - 15 v_{13} v_{2} - \frac{135}{4} v_{11} v_{22}\\
\hphantom{\cF_{40} =}{}
- \frac{45}{4} v_{2} v_{22} - 15 v_{1} v_{23} - 15 v_{12} v_{3} - \frac{405}{4} v_{1111111} - \frac{621}{8} v_{111112} - \frac{45}{4} v_{11113} - 15 v_{1123}\\
 \hphantom{\cF_{40} =}{}
- \frac{135}{8} v_{1222} - \frac{45}{4} v_{124} - 5 v_{133} - \frac{45}{4} v_{223} - 3 v_{25} - \frac{5}{2} v_{34} ,\\
\cdots
\end{gather*}
From $\cF_{10} = 0$ we get the equation
\begin{gather*} v_{22} = 12 v_{1} v_{11} - 3 v_{1111}. \end{gather*}
Using this equation we get $v_{23} = 0$ from $\cF_{20} = 0$. We take the liberty of integrating this without a constant and take $v_3 = 0$ as the second equation in the hierarchy. We can proceed iteratively and at each step integrate with respect to $t_2$ to find the hierarchy
\begin{gather*}
v_{22} = 12 v_{1} v_{11} - 3 v_{1111}, \\
v_{3} = 0, \\
v_{4} = -6 v_{1} v_{2} + 3 v_{112}, \\
v_{5} = -15 v_{1}^{3} + \frac{135}{4} v_{11}^{2} + 45 v_{1} v_{111} - \frac{15}{4} v_{2}^{2} - 9 v_{11111}, \\
\cdots
\end{gather*}
The first equation of the hierarchy is the potential Boussinesq equation. We observe that every third equation is trivial, $v_{3k} = 0$, as expected from the construction of the hierarchy using pseudodifferential operators \cite{dickey2003soliton}. Note that in this hierarchy all equations except the first are evolutionary. However, the higher equations are 3-dimensional: they contain $t_2$-derivatives that cannot be eliminated, because no evolutionary equation is available in the $t_2$-direction. It is surprising that a 2-dimensional lattice equation produces 3-dimensional equations in the continuum limit, but this is consistent with the limit procedure, because under the Miwa correspondence a~single lattice shift influences all continuous variables.

The discrete Lagrangian is
\begin{gather*}
L(U,U_1,U_2,U_{12},\lambda_1,\lambda_2) = \left( \frac{1}{\lambda_1^3} - \frac{1}{\lambda_2^3} \right) \log \left(\frac{1}{\lambda_1} - \frac{1}{\lambda_2} - U_1 + U_2 \right) \\
\hphantom{L(U,U_1,U_2,U_{12},\lambda_1,\lambda_2) =}{} - \left( \frac{1}{\lambda_1^2} + \frac{1}{\lambda_1 \lambda_2} + \frac{1}{\lambda_2^2} \right) (U_2 - U_1) + \left( \frac{1}{\lambda_1} + \frac{1}{\lambda_2} \right) (U_2 - U_1) U_{12} \\
\hphantom{L(U,U_1,U_2,U_{12},\lambda_1,\lambda_2) =}{} + \left(\frac{1}{\lambda_1} - \frac{1}{\lambda_2} + U_{2} - U_{1} \right) U U_{12} - \frac{1}{\lambda_1} U U_2 + \frac{1}{\lambda_2} U U_1.
\end{gather*}
In order to get the necessary leading order cancellation, we add the following terms to the Lagrangian
\begin{gather*}
 \frac{1}{2} \left( \frac{1}{\lambda_1} U U_1 + \frac{1}{\lambda_2} U_1 U_{12} - \frac{1}{\lambda_1} U_2 U_{12} - \frac{1}{\lambda_2} U U_ 2 \right)
+ \frac{1}{3}\big( U_1^3 - U_2^3 \big) \\
\qquad{} + \frac{1}{4}\left( \frac{1}{\lambda_1}\big(U_1^2-U^2\big) + \frac{1}{\lambda_2} \big(U_{12}^2-U_1^2\big) + \frac{1}{\lambda_1}\big(U_2^2-U_{12}^2\big) + \frac{1}{\lambda_2}\big(U^2-U_2^2\big) \right) .
\end{gather*}
These terms do not contribute to the Euler--Lagrange equations because they are the discrete exterior derivative of the 1-form
\begin{gather*} \eta(U,U_i,\lambda_i) = \frac{1}{2 \lambda_i} U U_i + \frac{1}{3} U^3 + \frac{1}{4 \lambda_i} \big(U_i^2 - U^2\big) . \end{gather*}

Some coefficients of the Lagrangian 2-form are
\begin{gather*}
\cL_{12} = 3 v_{1}^3 + \frac{9}{4} v_{11}^2 - \frac{3}{4} v_{2}^2,
\\
\cL_{13} = -\frac{3}{2} v_{2} v_{3},
\\
\cL_{14} = -9 v_{1}^4 - \frac{81}{2} v_{1} v_{11}^2 + 9 v_{} v_{12} v_{2} - \frac{27}{4} v_{111}^2 - \frac{9}{4} v_{12}^2 - \frac{3}{2} v_{2} v_{4},
\\
\cL_{15} = -\frac{45}{2} v_{1}^3 v_{2} + \frac{135}{4} v_{1}^2 v_{112} - \frac{135}{8} v_{11}^2 v_{2} - \frac{15}{8} v_{2}^3 + \frac{27}{2} v_{11} v_{1112} - \frac{3}{2} v_{2} v_{5},
\\
\cL_{23} = -9 v_{1}^2 v_{3} - \frac{9}{2} v_{11} v_{13} + \frac{9}{2} v_{111} v_{3},
\\
\cL_{24} = -36 v_{1}^3 v_{2} - 27 v_{1} v_{11} v_{12} - \frac{27}{2} v_{11}^2 v_{2} + 27 v_{1} v_{111} v_{2} + 3 v_{2}^3 + 9 v_{} v_{2} v_{22} - 9 v_{1}^2 v_{4} \\
\hphantom{\cL_{24} =}{}- \frac{27}{2} v_{111} v_{112} - \frac{9}{2} v_{11} v_{14} - \frac{9}{2} v_{12} v_{22} + \frac{9}{2} v_{111} v_{4},
\\
\cL_{25} = -54 v_{1}^5 + \frac{1215}{4} v_{1}^2 v_{11}^2 + \frac{135}{2} v_{1}^3 v_{111} - \frac{135}{4} v_{1}^2 v_{2}^2 + \frac{81}{8} v_{11}^2 v_{111} - \frac{81}{4} v_{1} v_{111}^2\!- 162 v_{1} v_{11} v_{1111}\\
\hphantom{\cL_{25} =}{} + \frac{135}{4} v_{1} v_{12}^2 + \frac{135}{4} v_{1}^2 v_{122} - \frac{135}{4} v_{11} v_{12} v_{2} + \frac{135}{8} v_{111} v_{2}^2- \frac{135}{2} v_{1} v_{11} v_{22}- 9 v_{1}^2 v_{5} \\
\hphantom{\cL_{25} =}{} + \frac{81}{4} v_{1111}^2 + \frac{27}{4} v_{112}^2 + \frac{27}{2} v_{11} v_{1122} - \frac{27}{2} v_{111} v_{122} - \frac{9}{2} v_{11} v_{15}
+ \frac{27}{2} v_{1111} v_{22} + \frac{9}{2} v_{111} v_{5},
\\
\cL_{34} = -36 v_{1}^3 v_{3} - 81 v_{1} v_{11} v_{13} + 9 v_{} v_{2} v_{23} + \frac{81}{2} v_{11}^2 v_{3} + 81 v_{1} v_{111} v_{3} - \frac{27}{2} v_{111} v_{113} + \frac{27}{2} v_{1111} v_{13}\\
\hphantom{\cL_{34} =}{} - \frac{9}{2} v_{12} v_{23} - \frac{27}{2} v_{11111} v_{3},
\\
\cL_{35} = -\frac{135}{2} v_{1}^2 v_{2} v_{3} + \frac{135}{4} v_{1}^2 v_{123} - \frac{135}{4} v_{11} v_{13} v_{2} - \frac{135}{2} v_{1} v_{11} v_{23} + \frac{135}{2} v_{1} v_{112} v_{3} \\
\hphantom{\cL_{35} =}{}+ \frac{135}{4} v_{11} v_{12} v_{3} + \frac{135}{4} v_{111} v_{2} v_{3} + \frac{27}{2} v_{11} v_{1123} - \frac{27}{2} v_{111} v_{123} + \frac{27}{2} v_{1112} v_{13} \\
\hphantom{\cL_{35} =}{}+ \frac{27}{2} v_{1111} v_{23} - \frac{27}{2} v_{11112} v_{3},
\\
\cL_{45} = 270 v_{1}^6 - 1215 v_{1}^3 v_{11}^2 - 1215 v_{1}^4 v_{111} - 135 v_{1}^3 v_{2}^2 + \frac{13851}{16} v_{11}^4 + \frac{8019}{4} v_{1} v_{11}^2 v_{111} \\
\hphantom{\cL_{45} =}{}+ \frac{6561}{4} v_{1}^2 v_{111}^2 + \frac{729}{2} v_{1}^2 v_{11} v_{1111} + \frac{405}{2} v_{1}^3 v_{11111} - \frac{405}{4} v_{1}^2 v_{12}^2 + 405 v_{1}^2 v_{112} v_{2}\\
\hphantom{\cL_{45} =}{} + \frac{405}{2} v_{1} v_{11} v_{12} v_{2} - \frac{2025}{8} v_{11}^2 v_{2}^2 - \frac{405}{4} v_{1} v_{111} v_{2}^2 + \frac{135}{16} v_{2}^4 - \frac{135}{2} v_{1}^2 v_{2} v_{4} + 36 v_{1}^3 v_{5}\\
\hphantom{\cL_{45} =}{} - \frac{81}{4} v_{111}^3 + \frac{243}{4} v_{11} v_{111} v_{1111} - \frac{243}{4} v_{1} v_{1111}^2 - \frac{3645}{8} v_{11}^2 v_{11111} - \frac{1215}{2} v_{1} v_{111} v_{11111} \\
\hphantom{\cL_{45} =}{}- \frac{567}{4} v_{1} v_{112}^2 + 81 v_{1} v_{1112} v_{12} - 162 v_{11} v_{112} v_{12} + \frac{81}{4} v_{111} v_{12}^2 + \frac{135}{4} v_{1}^2 v_{124} + 81 v_{1} v_{11} v_{15} \\
\hphantom{\cL_{45} =}{}- 81 v_{1} v_{11112} v_{2} + 81 v_{11} v_{1112} v_{2} - 81 v_{111} v_{112} v_{2} - \frac{81}{4} v_{1111} v_{12} v_{2} - \frac{135}{4} v_{11} v_{14} v_{2} \\
\hphantom{\cL_{45} =}{}+ \frac{405}{8} v_{11111} v_{2}^2 - \frac{135}{2} v_{1} v_{11} v_{24} - 9 v_{} v_{2} v_{25} + \frac{135}{2} v_{1} v_{112} v_{4} + \frac{135}{4} v_{11} v_{12} v_{4} \\
\hphantom{\cL_{45} =}{}+ \frac{135}{4} v_{111} v_{2} v_{4} - \frac{81}{2} v_{11}^2 v_{5} - 81 v_{1} v_{111} v_{5} + \frac{243}{4} v_{11111}^2 - \frac{81}{4} v_{1112}^2 + \frac{81}{2} v_{11112} v_{112}\\
\hphantom{\cL_{45} =}{} + \frac{27}{2} v_{11} v_{1124} + \frac{27}{2} v_{111} v_{115} - \frac{27}{2} v_{111} v_{124} + \frac{27}{2} v_{1112} v_{14} - \frac{27}{2} v_{1111} v_{15} + \frac{27}{2} v_{1111} v_{24} \\
\hphantom{\cL_{45} =}{}+ \frac{9}{2} v_{12} v_{25} - \frac{27}{2} v_{11112} v_{4} + \frac{27}{2} v_{11111} v_{5}.
\end{gather*}
Unfortunately, not all alien derivatives can be eliminated. Since we do not have any equation in the hierarchy to eliminate first derivatives with respect to $t_2$, the derivatives $v_2, v_{12}, \dots$ remain in place. All other derivatives are eliminated from the coefficients where they are alien.

\subsection{Continuum limit of lattice GD4}

In order to get a truly multicomponent pluri-Lagrangian system we move to the next equation of the GD hierarchy, with $N = 4$. In its 3-component form the discrete GD4 equation reads
\begin{gather}
 \left( \frac{1}{\lambda_1} - \frac{1}{\lambda_2} + U_{122} - U_{112} \right) V_{1122}
+ \left( \frac{1}{\lambda_1} - \frac{1}{\lambda_2} + U_{2} - U_{1} \right) W \nonumber\\
\qquad\quad{} - \frac{1}{\lambda_1} V_{122} + \frac{1}{\lambda_2} V_{112} + \frac{1}{\lambda_2} W_2 - \frac{1}{\lambda_1} W_1
- \left( V_{122} U_2 - V_{112} U_1 + U_{122} W_2 - U_{112} W_1 \right) \nonumber\\
 \qquad\quad{}- \left( \frac{1}{\lambda_1} + \frac{1}{\lambda_2} \right) \left( V_{122} - V_{112} - W_2 + W_1 \right) \nonumber\\
 \qquad {}= \frac{ \frac{1}{\lambda_1^4} - \frac{1}{\lambda_2^4} }{ \frac{1}{\lambda_1} - \frac{1}{\lambda_2} + U_{12} - U_{11} }
- \frac{ \frac{1}{\lambda_1^4} - \frac{1}{\lambda_2^4} }{ \frac{1}{\lambda_1} - \frac{1}{\lambda_2} + U_{22} - U_{12} }
- \left( \frac{1}{\lambda_1} + \frac{1}{\lambda_2} \right) \left( U_2 U_{122} + U_1 U_{112} \right) \nonumber\\
 \qquad\quad {}- \left( \frac{1}{\lambda_1^2} + \frac{1}{\lambda_1 \lambda_2} + \frac{1}{\lambda_2^2} \right)\left( U_{122} - U_{112} - U_2 + U_1 \right) ,\nonumber\\ 
V_2 - V_1 = \left( \frac{1}{\lambda_1} - \frac{1}{\lambda_2} + U_{2} - U_{1} \right) U_{12} - \frac{1}{\lambda_1} U_2 + \frac{1}{\lambda_2} U_1 ,\label{GD4-2} \\
W_2 - W_1 = -\left( \frac{1}{\lambda_1} - \frac{1}{\lambda_2} + U_{2} - U_{1} \right) U - \frac{1}{\lambda_2} U_2 + \frac{1}{\lambda_1} U_1 .\label{GD4-3}
\end{gather}
Equations \eqref{GD4-2} and \eqref{GD4-3} do not just look similar, also their expansions are nearly identical in leading order:
\begin{gather*}
(\lambda_2 - \lambda_1) v_1 = (\lambda_2 - \lambda_1) u_1 u + \frac{1}{2}(\lambda_2 - \lambda_1) u_{11} + \frac{1}{2} (\lambda_2 - \lambda_1) u_{2}
+ \cO\big( (\lambda_1+\lambda_2)^2 \big) , \\
(\lambda_2 - \lambda_1) w_1 = - (\lambda_2 - \lambda_1) u_1 u - \frac{1}{2}(\lambda_2 - \lambda_1) u_{11} + \frac{1}{2}(\lambda_2 - \lambda_1) u_{2}
+ \cO\big( (\lambda_1+\lambda_2 )^2 \big) ,
\end{gather*}
where we used the constant $c = 1$ in the Miwa correspondence~\eqref{miwa}. This gives us the ODE $v_1 - w_1 = 2 u_1 u + u_{11}$, which integrates to
\begin{gather}\label{GD4-constraint}
v - w = u^2 + u_1.
\end{gather}
We have omitted the integration constant because higher order terms in the expansion force it to be zero. equation~\eqref{GD4-constraint} allows us the eliminate either $v$ or $w$ from the continuous system. Hence we lose one of the components in the continuum limit. It is convenient to make a change of variables in the discrete system
\begin{gather*} V = \frac{X+Y}{2}, \qquad W = \frac{X-Y}{2}. \end{gather*}
Then in the continuum limit the variable $y$ can be eliminated by equation~\eqref{GD4-constraint}, which in the new variables reads
\begin{gather}\label{GD4-clim-constraint}
y = u^2 + u_1 .
\end{gather}
(One recognizes the Miura transformation~\cite{miura1968korteweg}.) For the remaining two variables we find the hierarchy
\begin{alignat}{3}
&u_2= x_1, \qquad && x_{22}= -4 u_{11} x_1 - 8 u_1 x_{11} - x_{1111}, &\notag \\
&u_3= \frac{3}{2} u_1^2 + \frac{1}{4} u_{111} +\frac{3}{4} x_2, \qquad && x_3= -3 u_1 x_1 - \frac{1}{2} x_{111},& \label{GD4-clim-hierarchy} \\
&u_4= 0, \qquad && x_4 = 0, &\notag \\
&\cdots & &\cdots &\notag
\end{alignat}
where, presumably, every fourth pair of equations is trivial.

The construction of the continuous GD4 hierarchy using pseudodifferential operators \cite{dickey2003soliton} yields the system
\begin{gather*}
\phi_2 = 2\chi_1 - 2\phi_{11}, \qquad
\chi_2 = \chi_{11} + 2\psi_1 - 2\phi_{111} - \phi \phi_1, \\
\psi_2 = \psi_{11} - \frac{1}{2}\phi_{1111} - \frac{1}{2}\phi \phi_{11} - \frac{1}{2}\chi \phi_1
\end{gather*}
as leading order equations. They are related to the leading order equations of the hierar\-chy~\eqref{GD4-clim-hierarchy} by the transformation
\begin{gather*} \phi = 4 u_1, \qquad \chi = 2 x_1 + 4 u_{11} . \end{gather*}
In pseudodifferential approach there is no obvious reason to eliminate one of the three variables. The reduction to two variables is forced upon us by the continuum limit, because there does not seem to be any way of performing the limit without getting an ODE relation between the variables, as in equation~\eqref{GD4-constraint}.

\subsubsection{Pluri-Lagrangian structure}

For $N = 4$ we have the Lagrangian
\begin{gather*}
L(U,V,W,U_1,V_1,W_1,U_2,V_2,W_2,U_{12},V_{12},W_{12},\lambda_1,\lambda_2) \\
= \left( \frac{1}{\lambda_1^4} - \frac{1}{\lambda_2^4} \right) \log \left(\frac{1}{\lambda_1} - \frac{1}{\lambda_2} - U_1 + U_2 \right) - \left( \frac{1}{\lambda_1^3} + \frac{1}{\lambda_1^2 \lambda_2} + \frac{1}{\lambda_1 \lambda_2^2} + \frac{1}{\lambda_2^3} \right) (U_2 - U_1) \\
\qquad {}+ \left( \frac{1}{\lambda_1^2} + \frac{1}{\lambda_1 \lambda_2} + \frac{1}{\lambda_2^2} \right) (U_2 - U_1) U_{12} - \left( \frac{1}{\lambda_1} + \frac{1}{\lambda_2} \right) (U_2 - U_1) V_{12} \\
\qquad {}- U \left( \left( \frac{1}{\lambda_1}- \frac{1}{\lambda_2} + U_2 - U_1 \right) V_{12} - \frac{1}{\lambda_1} V_2 + \frac{1}{\lambda_2} V_1 \right) \\
\qquad {}- \left( \left( \frac{1}{\lambda_1} + \frac{1}{\lambda_2} \right) U - W \right) \left( V_2 - V_1 - \left( \frac{1}{\lambda_1}- \frac{1}{\lambda_2} + U_2 - U_1 \right) U_{12} + \frac{1}{\lambda_1} U_2 - \frac{1}{\lambda_2} U_1 \right) .
\end{gather*}
As before, we make the change of variables
\begin{gather*} V = \frac{X+Y}{2}, \qquad W = \frac{X-Y}{2}. \end{gather*}
In order to achieve the necessary leading order cancellation, we add to the Lagrangian the discrete exterior derivative of the following discrete 1-form:
\begin{gather*}
\eta(U,X,Y,U_i,X_i,Y_i,\lambda_i) = \frac{1}{2 \lambda_i^2} U U_i- \frac{1}{4 \lambda_i^2} \big(U^2-U_i^2\big)+ \frac{1}{6 \lambda_i} \big(U U_i^2 + U^2 U_i\big) \\
\hphantom{\eta(U,X,Y,U_i,X_i,Y_i,\lambda_i) =}{}
- \frac{1}{2 \lambda_i} (U X_i + U_i X) + \frac{1}{2 \lambda_i} U X- \frac{1}{2 \lambda_i} U Y_i\\
\hphantom{\eta(U,X,Y,U_i,X_i,Y_i,\lambda_i) =}{}- \frac{1}{8} (Y Y_i - X X_i - X Y_i - X_i Y)- \frac{1}{4} \big(U_i^2 X + U^2 X_i\big) .
\end{gather*}
This 1-form was found by trial and error. It would be interesting to establish a general strategy to find a Lagrangian within a given equivalence class that provides the required cancellation. Without such a strategy, it seems infeasible to apply the continuum limit to higher members of the Gelfand--Dickey hierarchy.

The first few coefficients of the pluri-Lagrangian 2-form are given in Table~\ref{table-GD4-Lagrangian}. Note that the Lagrangian 2-form depends on all three fields, $u$, $x$, and $y$. The multi-time Euler--Lagrange equations are equivalent to equations~\eqref{GD4-clim-constraint} and~\eqref{GD4-clim-hierarchy}: they contain both the constraint on $y$ and the hierarchy in~$u$ and~$x$.

As with the Boussinesq hierarchy, the elimination of alien derivatives needs a comment. We do not have any equation in the hierarchy to eliminate the derivatives $x_2, x_{12}, \dots$, so these have to be tolerated in the coefficient $\cL_{13}$. All other alien derivatives, in particular those of $u$ and $y$, are eliminated as usual.

\newpage

\begin{table}[th!]
\fillbox{
\begin{gather*}
\cL_{12} = \tfrac{1}{3} u^{3} u_{12} - \tfrac{1}{3} u^{3} x_{11} - \tfrac{1}{3} u^{2} u_{112} - \tfrac{7}{6} u_{1}^{2} u_{2} + u_{1} u_{11} x - \tfrac{1}{2} u_{11} x y - \tfrac{1}{2} u_{1} x_{1} y + \tfrac{1}{2} u x y_{11}\\
\hphantom{\cL_{12} =}{} - \tfrac{1}{4} u^{2} y_{12} - \tfrac{1}{4} u_{1} u_{112} + \tfrac{1}{4} u_{112} x - \tfrac{1}{4} u_{1} x_{111} + \tfrac{1}{4} u x_{112} - \tfrac{1}{2} x x_{12} - u_{2} x_{2} - \tfrac{1}{4} x_{12} y \\
\hphantom{\cL_{12} =}{} - \tfrac{1}{4} x_{11} y_{1} + \tfrac{1}{4} u y_{112} + \tfrac{1}{4} x y_{12} \\
\cL_{13} = \tfrac{3}{2} u^{2} u_{1}^{3} + \tfrac{23}{12} u_{1}^{4} - \tfrac{1}{4} u^{2} u_{11}^{2} + \tfrac{1}{3} u^{3} u_{13} - \tfrac{3}{2} u_{1}^{3} x - 3 u u_{1} u_{11} x - \tfrac{1}{4} u^{3} x_{12} + \tfrac{3}{4} u_{1}^{3} y
\\
\hphantom{\cL_{13} =}{}
+ \tfrac{3}{2} u u_{1} u_{11} y - \tfrac{3}{2} u_{1} u_{11}^{2} - \tfrac{1}{3} u^{2} u_{113} - \tfrac{7}{6} u_{1}^{2} u_{3} + \tfrac{3}{4} u_{1} u_{11} x_{1} + 3 u x_{1} x_{11} + \tfrac{1}{4} u u_{1} x_{111}\\
\hphantom{\cL_{13} =}{}
- \tfrac{1}{4} u^{2} x_{13} + \tfrac{9}{8} u_{1}^{2} x_{2} + \tfrac{3}{8} u x_{12} y + \tfrac{3}{8} u_{1} x_{2} y - \tfrac{3}{8} u_{1} u_{11} y_{1} + \tfrac{3}{4} u x_{11} y_{1} + \tfrac{3}{4} u x_{1} y_{11}\\
\hphantom{\cL_{13} =}{}
- \tfrac{1}{8} u u_{1} y_{111} - \tfrac{1}{4} u^{2} y_{13} + \tfrac{1}{16} u_{111}^{2} - \tfrac{1}{4} u_{1} u_{113} + \tfrac{1}{4} u_{113} x - \tfrac{1}{8} u_{111} x_{11} + \tfrac{1}{4} x_{11}^{2} \\
\hphantom{\cL_{13} =}{}+ \tfrac{1}{16} u_{1} x_{112}
 + \tfrac{1}{2} u x_{113} + \tfrac{1}{2} x x_{13} - u_{3} x_{2} + \tfrac{3}{8} x_{2}^{2} + \tfrac{3}{16} x_{112} y + \tfrac{1}{16} u_{111} y_{11} - \tfrac{1}{8} x_{1} y_{111} \\
\hphantom{\cL_{13} =}{}+ \tfrac{1}{4} u y_{113} + \tfrac{1}{4} x y_{13} \\
\cL_{23} = \tfrac{3}{2} u^{3} u_{1} u_{12} + \tfrac{3}{4} u^{2} u_{1}^{2} u_{2} + \tfrac{3}{2} u u_{1}^{3} x - \tfrac{3}{2} u^{2} u_{1} u_{11} x - \tfrac{9}{4} u^{2} u_{1}^{2} x_{1} + \tfrac{3}{2} u^{3} u_{11} x_{1} + \tfrac{3}{2} u^{3} u_{1} x_{11}\\
\hphantom{\cL_{23} =}{}
+ \tfrac{7}{4} u u_{1}^{2} u_{12} - \tfrac{1}{2} u^{2} u_{11} u_{12} + \tfrac{14}{3} u_{1}^{3} u_{2} - \tfrac{5}{2} u u_{1} u_{11} u_{2} - \tfrac{3}{4} u_{1}^{2} u_{11} x + \tfrac{1}{4} u u_{1} u_{111} x\\
\hphantom{\cL_{23} =}{}
- \tfrac{1}{8} u^{2} u_{1111} x - 3 u u_{1} u_{12} x + \tfrac{1}{2} u^{2} u_{13} x - \tfrac{3}{2} u_{1}^{2} u_{2} x - u u_{1} u_{3} x - 5 u_{1}^{3} x_{1} + 2 u u_{1} u_{11} x_{1}\\
\hphantom{\cL_{23} =}{}
+ \tfrac{3}{4} u^{2} u_{111} x_{1} + 3 u u_{1} u_{2} x_{1} + u^{2} u_{3} x_{1} - 3 u u_{1} x_{1}^{2} + \tfrac{1}{4} u u_{1}^{2} x_{11} + \tfrac{5}{2} u^{2} u_{11} x_{11} + \tfrac{3}{4} u^{2} x_{1} x_{11}\\
\hphantom{\cL_{23} =}{}
+ u^{2} u_{1} x_{111} + \tfrac{3}{8} u^{3} x_{1111} + \tfrac{3}{4} u^{2} u_{1} x_{12} - \tfrac{3}{8} u^{2} x x_{12} + \tfrac{1}{3} u^{3} x_{13} - \tfrac{3}{2} u u_{1}^{2} x_{2} - \tfrac{3}{8} u^{2} u_{2} x_{2}\\
\hphantom{\cL_{23} =}{}
+ \tfrac{3}{4} u u_{1} x x_{2} - \tfrac{3}{8} u^{2} x_{1} x_{2} + \tfrac{1}{8} u^{3} x_{22} - \tfrac{1}{2} u^{2} u_{1} x_{3} + \tfrac{3}{2} u_{1} u_{11} x y - \tfrac{3}{4} u_{1}^{2} x_{1} y - \tfrac{3}{2} u u_{11} x_{1} y\\
\hphantom{\cL_{23} =}{}
- \tfrac{3}{2} u u_{1} x_{11} y - \tfrac{3}{4} u_{1}^{2} x y_{1} + \tfrac{3}{2} u u_{1} x_{1} y_{1} + \tfrac{3}{4} u u_{1}^{2} y_{2} + \tfrac{1}{8} u_{1}^{2} u_{112} + \tfrac{1}{8} u u_{11} u_{112} - 3 u_{1} u_{11} u_{12}\\
\hphantom{\cL_{23} =}{}
+ \tfrac{1}{24} u u_{111} u_{12} + \tfrac{3}{8} u_{11}^{2} u_{2} + \tfrac{13}{8} u_{1} u_{111} u_{2} - \tfrac{1}{6} u u_{1111} u_{2} + \tfrac{1}{8} u_{11} u_{111} x - \tfrac{1}{8} u_{1} u_{1111} x\\
\hphantom{\cL_{23} =}{}
- \tfrac{1}{2} u_{1} u_{13} x - \tfrac{1}{2} u_{11} u_{3} x - \tfrac{7}{8} u_{11}^{2} x_{1} - \tfrac{7}{8} u_{1} u_{111} x_{1} + \tfrac{1}{6} u u_{1111} x_{1} + \tfrac{3}{2} u_{1} u_{12} x_{1} - \tfrac{2}{3} u u_{13} x_{1}\\
\hphantom{\cL_{23} =}{}
- \tfrac{3}{4} u_{11} u_{2} x_{1} + 4 u_{1} u_{3} x_{1} + \tfrac{3}{4} u_{11} x_{1}^{2} + x_{1}^{3} + \tfrac{7}{2} u_{1} u_{11} x_{11} - \tfrac{1}{24} u u_{111} x_{11} - \tfrac{3}{2} u_{1} u_{2} x_{11}\\
\hphantom{\cL_{23} =}{}
+ \tfrac{2}{3} u u_{3} x_{11} + \tfrac{3}{4} u_{1} x_{1} x_{11} + \tfrac{5}{8} u_{1}^{2} x_{111} - \tfrac{3}{8} u u_{11} x_{111} + \tfrac{1}{4} u u_{2} x_{111} - \tfrac{1}{4} u x_{1} x_{111} \\
\hphantom{\cL_{23} =}{}+ \tfrac{1}{4} u u_{1} x_{1111}
+ \tfrac{3}{16} u^{2} x_{11111} + \tfrac{1}{8} u^{2} x_{1112} + \tfrac{1}{2} u u_{12} x_{2} + \tfrac{3}{4} u_{1}^{2} x_{12} - \tfrac{1}{2} u u_{2} x_{12} - \tfrac{3}{8} u_{1} x x_{12}\\
\hphantom{\cL_{23} =}{} + \tfrac{7}{2} u x_{1} x_{12}
+ \tfrac{3}{16} u^{2} x_{122} + \tfrac{1}{2} u u_{1} x_{13} + \tfrac{3}{8} u_{1} u_{2} x_{2} + \tfrac{3}{8} u_{11} x x_{2} - \tfrac{27}{8} u_{1} x_{1} x_{2} - \tfrac{1}{2} u x_{11} x_{2}\\
\hphantom{\cL_{23} =}{}
- \tfrac{1}{2} u u_{11} x_{3}- \tfrac{1}{2} u x_{1} x_{3} + \tfrac{1}{8} u_{1111} x y - \tfrac{3}{4} u_{111} x_{1} y - \tfrac{9}{4} u_{11} x_{11} y - \tfrac{3}{4} x_{1} x_{11} y - \tfrac{7}{4} u_{1} x_{111} y \\
\hphantom{\cL_{23} =}{}- \tfrac{3}{8} u x_{1111} y
- \tfrac{3}{4} u_{1} x_{12} y + \tfrac{3}{8} x x_{12} y + \tfrac{3}{8} x_{1} x_{2} y - \tfrac{1}{8} u_{111} x y_{1} + \tfrac{1}{8} u_{11} x_{1} y_{1} + \tfrac{3}{4} x_{1}^{2} y_{1} \\
\hphantom{\cL_{23} =}{}+ \tfrac{1}{4} u_{1} x_{11} y_{1}
+ \tfrac{3}{8} u x_{111} y_{1} + \tfrac{3}{4} u x_{12} y_{1} - \tfrac{3}{8} x x_{2} y_{1} - \tfrac{1}{8} u_{1} x_{1} y_{11} - \tfrac{1}{8} u x_{11} y_{11} - \tfrac{1}{8} u u_{1} y_{112}\\
\hphantom{\cL_{23} =}{} - \tfrac{1}{4} u_{1}^{2} y_{12}
+ \tfrac{1}{8} u u_{11} y_{12} + \tfrac{3}{4} u x_{1} y_{12} - \tfrac{1}{2} u x y_{13} + \tfrac{3}{8} u x_{2} y_{2} + \tfrac{1}{2} u_{1} x y_{3} + \tfrac{3}{16} u_{111} u_{112}\\
\hphantom{\cL_{23} =}{} - \tfrac{1}{8} u_{1111} u_{12}
+ \tfrac{1}{16} u_{11111} u_{2} - \tfrac{1}{16} u_{11111} x_{1} + \tfrac{1}{8} u_{1111} x_{11} - \tfrac{1}{8} u_{112} x_{11} - \tfrac{1}{4} u_{13} x_{11} \\
\hphantom{\cL_{23} =}{}- \tfrac{3}{16} u_{111} x_{111}
+ \tfrac{1}{4} u_{12} x_{111} + \tfrac{3}{4} u_{3} x_{111} - \tfrac{1}{4} x_{11} x_{111} - \tfrac{1}{16} u_{11} x_{1111} - \tfrac{1}{4} u_{2} x_{1111}\\
\hphantom{\cL_{23} =}{}+ \tfrac{1}{4} x_{1} x_{1111} + \tfrac{3}{16} u_{1} x_{11111}
+ \tfrac{1}{8} u_{1} x_{1112} - \tfrac{1}{8} u_{11} x_{112} - \tfrac{5}{16} u_{2} x_{112} + \tfrac{5}{16} x_{1} x_{112} + \tfrac{1}{4} u_{1} x_{113} \\
\hphantom{\cL_{23} =}{}+ \tfrac{3}{16} u_{12} x_{12} - \tfrac{1}{4} u_{3} x_{12}
+ \tfrac{3}{8} x_{11} x_{12} + \tfrac{1}{4} u_{1} x_{122} + \tfrac{1}{4} u x_{123} + \tfrac{1}{2} x_{1} x_{13} + \tfrac{1}{4} u_{13} x_{2}\\
\hphantom{\cL_{23} =}{} - \tfrac{9}{16} x_{111} x_{2} - \tfrac{1}{4} u_{11} x_{22} + x x_{23}
- \tfrac{1}{2} x_{11} x_{3} + \tfrac{1}{2} x_{2} x_{3} - \tfrac{3}{16} x_{11111} y - \tfrac{1}{8} x_{1112} y \\
\hphantom{\cL_{23} =}{}+ \tfrac{1}{16} x_{1111} y_{1} + \tfrac{1}{8} x_{112} y_{1} - \tfrac{1}{8} x_{12} y_{11}
- \tfrac{1}{8} x_{1} y_{112} + \tfrac{1}{8} x_{11} y_{12} - \tfrac{3}{16} x_{2} y_{12} + \tfrac{1}{4} x_{1} y_{13}\\
\hphantom{\cL_{23} =}{} + \tfrac{3}{16} x_{12} y_{2} + \tfrac{1}{4} x_{2} y_{3}
\end{gather*}
}
\caption{Coefficients of the Lagrangian 2-form for the 4th Gelfand--Dickey hierarchy.}\label{table-GD4-Lagrangian}
\end{table}

\newpage
\section{Conclusions}

The continuum limit procedure developed in \cite{vermeeren2019continuum} can be applied to many discrete pluri-Lag\-ran\-gian systems. Starting from equations from the ABS list and from discrete Gelfand--Dickey equations, we obtained pluri-Lagrangian structures for among others the Krichever--Novikov, Korteweg--de Vries, and Gelfand--Dickey hierarchies. For many of these hierarchies this is the first time a~pluri-Lagrangian structure has been given. The success rate of our method suggests that it is a~useful tool to study connections between discrete and continuous integrable systems. Nevertheless, there are lattice systems to which the continuum limit procedure does not seem to apply, an issue which deserves additional investigation. Previously, only a few examples of continuous pluri-Lagrangian 2-form systems were known. The multitude of new examples given in this work contributes to the evidence that a pluri-Lagrangian structure is a relevant attribute of integrability.

\subsection*{Acknowledgments}

This research was supported by the DFG through the SFB/TRR 109, `Discretization in Geo\-metry and Dynamics'. The author is grateful to the anonymous referees for their insightful comments.

\pdfbookmark[1]{References}{ref}
\LastPageEnding

\end{document}